\documentclass[USenglish,oneside,twocolumn]{article}

\usepackage[utf8]{inputenc}%(only for the pdftex engine)

\usepackage{amsfonts,amsmath,amssymb,amsthm}
\usepackage{algorithmic}
\usepackage{graphicx}
\usepackage{textcomp}
\usepackage{xcolor}
\usepackage{listings}
\usepackage{hyperref}
\usepackage{cleveref}
\usepackage{upquote}
\usepackage{flushend}
\usepackage{booktabs}

\usepackage{tikz}
\usepackage{pgfplotstable}
\usepackage{pgfplots}
\usepackage{tabularx}
\usepackage{microtype}

\theoremstyle{theorem}
\newtheorem{theorem}{Theorem}
\theoremstyle{theorem}
\newtheorem{lemma}{Lemma}
\theoremstyle{definition}
\newtheorem{definition}[theorem]{Definition}
\theoremstyle{definition} 

\theoremstyle{definition} 

\newcommand{\ignore}[1]{}
 
\begin{document}

\title{Privacy-Preserving Payment Splitting}
  \author{Saba Eskandarian\thanks{Stanford University. Email:~\texttt{saba@cs.stanford.edu}.} \and 
			Mihai Christodorescu\thanks{Visa Research.} \and
			Payman Mohassel\thanks{Facebook (work done while at Visa Research).} 
}
  \date{}
  %\subtitle{...}
\maketitle

  \begin{abstract}
Widely used payment splitting apps allow members of a group to keep track of debts between members by sending charges for expenses paid by one member on behalf of others. While offering a great deal of convenience, these apps gain access to sensitive data on users' financial transactions. In this paper, we present a payment splitting app that hides all transaction data within a group from the service provider, provides privacy protections between users in a group, and provides integrity against malicious users or even a malicious server. 

The core protocol proceeds in a series of rounds in which users either submit real data or cover traffic, and the server blindly updates balances, informs users of charges, and computes integrity checks on user-submitted data. Our protocol requires no cryptographic operations on the server, and after a group's initial setup, the only cryptographic tool users need is AES. 

We implement the payment splitting protocol as an Android app and the accompanying server. We find that, for realistic group sizes, it requires fewer than 50 milliseconds per round of computation on a user's phone and the server requires fewer than 300 \emph{micro}seconds per round for each group, meaning that our protocol enjoys excellent performance and scalability properties.
\end{abstract}

\section{Introduction}

Payment-splitting apps solve the problem of keeping track of debts between members of a group. They provide a convenient interface and bookkeeping system for groups to keep track of individual or communal costs among their members. Common use cases involve splitting bills for meals among friends or colleagues, roommates keeping track of grocery, utility, or rent charges, and groups of travelers monitoring expenses during a trip. Groups can exist for a short term, e.g. a vacation group, or indefinitely. There exist a great number and variety of free payment splitting apps\footnote{Some popular payment splitting apps: Splitwise, Receipt Ninja, BillPin, SpotMe, Conmigo, and Settle Up.}, with multiple having over 1 million downloads and tens of thousands of positive reviews on the Google Play app store. 

Unfortunately, these apps leak a great deal of private user data to the service provider. The privacy policy for Splitwise~\cite{splitwiseprivacy}, perhaps the most well-known payment splitting app, permits collection of, ``for example, group names, expense descriptions and amounts, payments and their confirmation numbers, comments and reminders, receipt images, notes, and memos, in addition to any other information that you attach or share'' as well as ``the types of expenses you add, the features you use, the actions you take, and the time, frequency and duration of your activities.'' 

A recent survey covering person-to-person payment preferences in the US~\cite{cardtronics} discovered that three times as many people prefer cash compared to mobile apps for person-to-person payments. Among the perceived benefits of cash, privacy leads in importance by the largest margin, even when considering the convenience of mobile payments. This suggests a need for a system that combines the convenience of mobile payment splitting apps with cash-like privacy. Although generic techniques from fully homomorphic encryption~\cite{Gentry09,BGV11,GSW13} and server-aided multiparty computation~\cite{FKN94,KMR11,KMR12,HLP11} or even a system built on top of %metadata-hiding anonymous messaging protocols~\cite{riposte,vuvuzela,stadium,pung,atom} or 
privacy-preserving cryptocurrencies~\cite{mw,conftrans,zerocash} could theoretically solve this problem, these solutions fall short of the need for an efficient, highly scalable, and secure solution. 

In this work, we present a payment splitting app that hides all transaction data within a group from the service provider while also providing privacy against non-involved group members and integrity against malicious users or even a malicious server. Our payment splitting protocol hides which group member pays in a transaction, which member is paid, how much is spent, when the transaction occurs, and even whether a transaction has happened at all from the server and entities external to the group. The server learns only the membership of each group and nothing more. Within a group, users not involved in a given real-world transaction will not learn which group members are indebted by that transaction. Moreover, our protocol is designed to ensure that neither a malicious user nor the server itself can alter balances, frame other users for charges, or otherwise tamper with the system's operation. Its server-side computation requires no cryptographic operations, only arithmetic on 128-bit integers, and, after an initial setup phase, the only cryptographic operations on user devices are evaluations of AES. A small-scale user survey motivates design decisions regarding our security requirements and the construction of the final system.

Our protocol operates in a setting where users can connect to a central server but cannot communicate directly with each other, a common design among existing (non-private) payment splitting apps. The server facilitates group setup, and clients then communicate through a server-assisted protocol in a series of rounds. Members of each group share a secret key. In each round, users send either a transaction or cover traffic to the server. At a high level, each user's message is a vector containing masked transaction information. We carefully design the structure and content of messages to allow the server to blindly update balances, to allow users to reject incorrect charges, to check that no user has attempted to tamper unfairly with its own or others' balances, and to minimize information learned by users not involved in a transaction. Moreover, our protocol takes advantage of the structure of the payment splitting problem to provide integrity protections without relying on zero-knowledge proof techniques.

We implement an Android app and an accompanying server that run our protocol and evaluate them on commodity hardware. We find that, for groups of up to 25 members (a realistic size, as determined by our survey), our implementation requires fewer than 50 milliseconds per round of computation on a user's phone and the server requires fewer than 300 \emph{micro}seconds per round for each group, with most transactions requiring only a few rounds to complete at most. Bandwidth requirements are also light, at under 500~bytes of communication between a client and the server in each round. 

Our work is, to our knowledge, the first to consider the problem of privacy-preserving payment splitting apps. We demonstrate that other, general-purpose solutions cannot be easily adapted to efficiently solve this problem by comparing our system to ZKLedger~\cite{zkledger}, a system designed for privacy-preserving audits of distributed ledgers which can be adapted to serve as a payment splitting system. Our protocol's computation time and bandwidth usage outperform ZKLedger by 7.3$\times$ and 6.8$\times$ respectively if it were used for a 10 member payment splitting group, the largest size for which ZKLedger reports end to end transaction times.

In summary, we make the following contributions:
\begin{itemize}
\item We introduce the notion of a privacy-preserving payment splitting scheme and provide formal definitions that model security for such a system.

\item We construct and prove the security of a practical and scalable privacy-preserving payment splitting scheme.

\item We implement and evaluate our system as an Android app and an accompanying server, finding that it requires fewer than 300$\mu$s of computation per group on the server and 50ms on the client for each round of the protocol in realistically sized groups.
\end{itemize}

\section{System Goals}\label{secmodel}

In this section, we present the goals for our system as well as the notation and security definitions required of a privacy-preserving payment splitting scheme. 

\vspace{.3em}\noindent\textbf{User Survey}. Before building our system, we conducted a short survey in order to understand how payment splitting apps are used and to guide our system design and evaluation. The survey was sent via email to approximately 250 individuals (employees in an office building belonging to a large company) of which 51 responded. All institutional requirements to administer this survey were satisfied before it was distributed. The full text of the survey appears in Appendix~\ref{survey}. 

In terms of privacy preferences, we found that only 4\% of respondents prefer to have their transactions be public, which we take as confirmation of our belief in the importance of investigating privacy-preserving payment splitting. 70\% of respondents preferred for transactions to be visible only to their participants, motivating the inclusion of debtor privacy in our security definitions. The data considered most sensitive about transactions was the identity of the parties involved. We will discuss other aspects of the survey responses as they become relevant to the design of our solution. 

\subsection{Functionality Goals}
A payment splitting scheme, as we define it, allows users to establish a group, add/remove members, send each other charges, reject unwanted charges, and settle the group balance. These operations can be categorized into group management operations, such as setting up a group and adding/removing users, and payment splitting operations that deal with the real functionality offered by the scheme. Our definitions and security arguments primarily focus on modeling payment splitting operations, as the group management operations are fairly straightforward and do not impact security. We support the following payment splitting operations.

\noindent\textbf{Request}. Any user in a group may send a request for any amount of money to any other user. 

\noindent\textbf{Reject}. Since a payment splitting app cannot know what real-world payments actually occured, users can reject charges which they dispute or were sent in error. 

Splitwise and some of its competitors immediately subtract money from a user's balance when the user is sent a request and then restore the money if the request happens to be rejected. As such, we observe that any system which offers the ability to trace who has made a given payment request automatically also implicitly allows rejections because an identical payment request can be sent in the opposite direction to cancel the first. For this reason, we focus on including a \emph{trace} functionality that reveals who initiated a request instead of implementing rejection directly. With the proper application logic running on top of the core protocol (as in our system), this approach provides the same interface to the user of a private payment splitting app as one that has a built-in rejection operation.

\noindent\textbf{Settle}. Payment splitting apps simplify paying friends back by treating users' debts as being owed to the group instead of to various individuals. While users can pay each other back directly for each charge, the service can also simplify payments by telling each user who to pay and how much to pay in order to most efficiently settle all the group's balances at once. 

One way to implement a settle operation would be for users to get a matrix indicating who should pay whom. We opt to describe this in terms of users getting a vector of balances because such a matrix can easily be derived from the vector. 

Having described the payment splitting functionality we plan to achieve, we end this section with a correctness property that our payment splitting scheme must satisfy. The definition states the intuitive notion that, as long as all parties follow the protocol, the scheme should maintain an accurate running balance of debts between group members after each operation in a way that the originator of each transaction can be traced. 

\begin{definition}[Correctness]
A payment splitting scheme is \emph{correct} if, when all users and the server act honestly,
\begin{itemize}
\item[i)] the balances revealed by a settle operation correspond to the sum of all requests made since the formation of the group when the vector of balances was filled with zeros, and \item[ii)] after each transaction, each user's own recorded balance always matches the corresponding value in the vector that would be returned by a settle operation. 
\end{itemize}
\end{definition}

\subsection{Security Goals} Our system achieves the following security properties, summarized below and described in greater depth later in this section.

\noindent\textbf{Server Privacy}. For any transaction, the server should not know which group members are involved in the transaction or how much is spent in the transaction. Details of who makes transactions when, rejected charges, etc should also be hidden.

\noindent\textbf{Debtor Privacy}. Up until the group settles (when it becomes clear who is indebted because they need to pay), no transaction will reveal which users it puts into debt. 

\noindent\textbf{User Integrity}. No user should be able to take advantage of privacy to create money for him/herself or otherwise maliciously tamper with balances.

\noindent\textbf{Server Integrity}. A malicious server cannot cause the protocol to deviate from correct functionality except by denial of service. This requirement only applies in the malicious security setting and is not included when we consider a semihonest server that follows the rules of the protocol while trying to learn user secrets. 

Note that although our system will provide privacy and integrity guarantess against both malicious users and servers, it will not protect against collusion between a malicious user and the server. Fortunately this combination of malicious actors is not a major concern in practice since the threats posed by a malicious server and a malicious user are usually orthogonal. For example, payment splitting groups typically consist of people who know each other in real life, making it infeasible for a malicious server to insert a fake user into a group so long as the group setup procedure is secure (we discuss options for setting up groups in our scheme in Section~\ref{overviewSec}). On the other hand, the owner of a payment splitting server is not typically invited into a group by the users of the service, so a malicious server operator would need to find another way to compromise integrity. One shortcoming of this approach is that it prevents a member of a group from self-hosting a payment splitting server, but self-hosting is not supported by most (non-private) payment splitting apps in use today either. Supporting security against malicious clients and servers that collude is a compelling problem for future work. 

\subsection{Security Definitions}

We now define the security requirements our system must meet. In terms of server privacy, we want our scheme to reveal to the server only group membership and hide any details of transaction parties, quantities, frequencies, etc. We formalize this with a security game where no server can distinguish between two potential transcripts of requests for a given group. Our server privacy definition also implies protection against an adversary who eavesdrops on the network or controls users in other groups in addition to the server itself.

\begin{definition}[Server Privacy Experiment]
The server privacy experiment \textsf{PRIV[$\mathcal{A}$, $\lambda$, b]} with security parameter $\lambda$ is played between an adversary $\mathcal{A}$ who plays the role of the server and a challenger $\mathcal{C}$ who is given input $b$ and plays the honest users $\mathcal{U}_1...\mathcal{U}_N$. 
\begin{enumerate}
\item Setup. Adversary $\mathcal{A}$ picks a group size $N$ and sets up a group with $\mathcal{C}$ playing the role of the users. 

\item Transactions. $\mathcal{A}$ sends $\mathcal{C}$ two transactions $t_0$ and $t_1$, both of which are either a request between group members or a settle operation. If exactly one of $t_0$ and $t_1$ is a settle operation, $\mathcal{C}$ aborts the experiment and outputs 0. Next, $\mathcal{C}$ interacts with $\mathcal{A}$ to carry out operation $t_b$. $\mathcal{A}$ can repeat this step as many times as it wishes. 

\item Output. $\mathcal{A}$ outputs a bit $b'$. 
\end{enumerate}

\textsf{PRIV[$\mathcal{A}$, $\lambda$, b]} outputs the value $b'$ returned by $\mathcal{A}$ at the end of the game. 
\end{definition}

\begin{definition}[Server Private]
A payment splitting scheme is \emph{server private} if no PPT adversary can win the server privacy game with greater than negligible advantage. That is, if the quantity $|\textsf{Pr}[\textsf{PRIV}[\mathcal{A}, \lambda, 0]=1]-\textsf{Pr}[\textsf{PRIV}[\mathcal{A}, \lambda, 1]=1]\leq\textsf{negl}(\lambda)|$ for any PPT $\mathcal{A}$.
\end{definition} 

In addition to complete privacy against the server, we require a notion of privacy against other users in a group as well, which we call \emph{debtor privacy}. Debtor privacy protects the privacy of users who become indebted to other group members by hiding the target of any request. Informally, we say that a payment splitting scheme is \emph{debtor private} if any coalition of compromised users cannot determine the identity of the requestee in a given transaction before the group settles. Note that since the adversary in this game corrupts members of the group, it has access to any group-wide secrets used to hide information from the server. 

\begin{definition}[Debtor Privacy Experiment]
The debtor privacy experiment \textsf{DEBTPRIV[$\mathcal{A}$, $\lambda$, b]} with security parameter $\lambda$ is played between an adversary $\mathcal{A}$ who plays the role of compromised users, and a challenger $\mathcal{C}$ who plays the server and uncompromised users and is given input $b$. 
\begin{enumerate}
\item Setup. Adversary $\mathcal{A}$ picks a group size $N$ and a set of indexes $M\subset [N]$ of users to corrupt, $|M|=n, n < N$. Adversary $\mathcal{A}$ and challenger $\mathcal{C}$ set up a group with $\mathcal{A}$ playing the role of users $\mathcal{U}_i$ for $i\in M'$ where $M'\subseteq M$ is a subset of the adversary's choosing. Challenger $\mathcal{C}$ plays the role of the server and all users $\mathcal{U}_i, i\notin M'$. After the group is set up, $\mathcal{C}$ sends $\mathcal{A}$ all secrets and any other state it holds for users $\mathcal{U}_i$, $i\in M$.

\item Transactions. $\mathcal{A}$ sends $\mathcal{C}$ two transactions $t_0$ and $t_1$, both of which are either a request between group members or a settle operation. If any of the following conditions are met, $\mathcal{C}$ aborts the experiment and outputs~0. 

\begin{enumerate}
\item Exactly one of $t_0$ and $t_1$ is a settle operation. 

\item A user $\mathcal{U}_i, i\in M$ is either making a charge or being charged and $t_0\neq t_1$. 

\item A user $\mathcal{U}_i$ making a request differs between $t_0$ and $t_1$.

\item $t_0$ and $t_1$ are settle operations, but the balances returned differ for any user. 
\end{enumerate}

Next, $\mathcal{C}$ interacts with $\mathcal{A}$ to carry out operation $t_b$. $\mathcal{A}$ can repeat this step as many times as it wishes. 

\item Output. $\mathcal{A}$ outputs a bit $b'$. 

\end{enumerate}

\textsf{DEBTPRIV[$\mathcal{A}$, $\lambda$, b]} outputs the value $b'$ returned by $\mathcal{A}$ at the end of the game. 
\end{definition}

\begin{definition}[Debtor Private]
A payment splitting scheme is \emph{debtor private} if no PPT adversary can win the debtor privacy game with greater than negligible advantage. That is, if the quantity $|\textsf{Pr}[\textsf{DEBTPRIV}[\mathcal{A}, \lambda, 0]=1]-\textsf{Pr}[\textsf{DEBTPRIV}[\mathcal{A}, \lambda, 1]=1]\leq\textsf{negl}(\lambda)|$ for any PPT $\mathcal{A}$.
\end{definition} 

Our definition of debtor privacy only applies in the context of requests and not settlements. Since money must change hands outside of the system when a group settles, it is necessary that some information about who is in debt be revealed at that point. 

One could imagine even stronger notions of privacy against other users, but we leave the task of building stronger security notions that precisely quantify leakage in settlement and also provide privacy for the requester to future work. We feel that protecting debtors provides a good balance between security and functionality where critical privacy needs are addressed while the resulting system can still be built from the most lightweight cryptographic tools.

We also require groups to have integrity, which we separate into \emph{user integrity} and \emph{server integrity}. User integrity requires three properties. First, we want to ensure that users cannot silently corrupt the balances kept by the server. We capture this property by observing that the balances kept by the server are valid so long as they all sum to zero, meaning that if everyone who is in debt pays, then everyone who is owed money gets all their money back. 
Second, we need to ensure that even if some users are malicious, they cannot ``confuse'' other users by causing them to have a locally stored balance that differs from their reported balance if the group were to settle. 
Third, we want to prevent malicious charges between users, but a malicious charge (i.e., a charge disputed by a user, though processed correctly by the system) can simply be rejected or charged back by the user who does not want to pay (just like in payment apps widely in use today). We must ensure, however, that an attacker cannot get away with framing a different group member as the originator of an unwanted charge to avoid being charged back. We include both properties in the definition of user integrity below. Section~\ref{extensions} describes how we can achieve stronger notions of user integrity, e.g., where we identify which user attempted to corrupt the sum of user balances or allow a framed user to not only detect but also prove that she has been framed. 

\begin{definition}[User Integrity]
Consider a PPT adversary $\mathcal{A}$ who corrupts up to $N-1$ users $\mathcal{U}_i$ in a group of size $N$. We say that a payment splitting scheme has \emph{user integrity} if the following properties are satisfied except with negligible probability in the security parameter~$\lambda$:

\begin{itemize}
\item[--] After each request, either the server detects it as an invalid transaction (and can roll it back) or, if the members of the group were to settle at that time, the resulting vector of user balances would sum to zero.

\item[--] After each request, either the server detects it as an invalid transaction (and can roll it back) or, if the members of the group were to settle at that time, each honest user's entry in the vector returned by settling would match its locally stored balance. 

\item[--] Any attempt to tamper with the output of the trace functionality such that it falsely points to an honest user who was not involved in a transaction can be detected by the framed honest user.

\end{itemize}
\end{definition}

Finally, server integrity requires that a malicious server cannot tamper with user balances without being detected by the users. 

\begin{definition}[Server Integrity]
Consider a (potentially malicious) PPT server $\mathcal{S}^*$ operating a group of size $N$. We say that a payment splitting scheme has \emph{server integrity} if, after each transaction, if members of the group were to settle at that point, they would either detect that $\mathcal{S}^*$ has acted maliciously or output the same vector of balances as they would when interacting with an honest server $\mathcal{S}$, except with negligible probability in the security parameter~$\lambda$.
\end{definition}

\section{Architecture Overview}\label{arch}
\ignore{
\begin{figure}
\centering
\includegraphics[width=.9\linewidth]{arch.png}
\caption{In our architecture, the server keeps track of user balances without ever seeing the balances or knowing who is making charges. The figure depicts a transaction where user 1 indicates that user 3 owes him money.}
\label{architecture}
\end{figure}
}

Our architecture consists of a mobile app and a server operating in the setting where devices running the mobile app have a secure network connection with the server but no connection with each other, as is commonly the case in payment apps used today. In addition to being widely used in practice, this architecture enables convenient group management and enjoys faster latency than decentralized systems.
We leave the investigation of alternative settings, e.g. peer-to-peer distributed networks between phones, as an interesting problem for future work and briefly discuss some possibilities in Section~\ref{discussion}.

Users of our app organize themselves into groups, and users within a group can send charges to each other to request money. Different groups operate independently of each other, and one server can support many groups at once. We describe our solution in the context of a single group, but the protocol can be repeated separately in parallel for as many groups as a server can support.

Similar to many other privacy-preserving protocols (for example,~\cite{anonrep,privateride,riposte,stadium,pung}), our core protocol proceeds in a series of rounds. 
To hide which users are and are not participating in transactions, users in each round send the server either a message representing a transaction or cover traffic to hide real transactions.
With all users online, this provides complete anonymity within the group for the user sending a real transaction (provided the protocol used provides server privacy). We discuss the resilience of our solution's anonymity to users going offline in Section~\ref{extensions}. 

\vspace{.3em}\noindent\textbf{Inadequacy of Trivial Solutions}. The rest of this section discusses a number of solutions to the problem of privacy-preserving payment splitting that use powerful generic ideas from cryptography or attempt to use simple tools na\"ively. We sketch each approach and then explain why it is inadequate either in terms of performance or security. 

The most generic cryptographic tools available for this problem come from fully homomorphic encryption (FHE)~\cite{Gentry09,BGV11,GSW13} or server-aided multiparty computation~\cite{FKN94,KMR11,KMR12,HLP11}. These techniques allow users to upload ciphertexts to a server who can then compute arbitrary functions on the encrypted data and send users back the result. Unfortunately, techniques from fully homomorphic encryption remain too slow for use in all but the most limited settings. Although adding subtracting from a user's balance only requires support for addition and subtraction, which can be achieved with significantly more lightweight cryptographic tools, one of the core technical challenges of our work lies in finding ways to allow a server to blindly \emph{route} payments between users without the full power of general FHE. 
Another possibility is to use multiparty computation, but multiparty computation techniques in the server-aided setting rely on garbled circuits~\cite{Yao82b,Yao86}, which operate on boolean circuits and therefore incur an additional evaluation of AES \emph{for each gate} in the boolean circuit representing the function to be calculated. We further discuss related work in generic techniques applicable to payment splitting in Section~\ref{related}.

Instead of using powerful cryptographic tools, one may also try to achieve our security goals through na\"ive use of basic cryptographic tools such as encryption or signatures. Consider the trivial scheme where users simply broadcast their encrypted transactions to all other users in a group, using the server to route each message to all other group members. In order to ensure user integrity, users could then gossip the messages they receive by re-sending them to all other group members. This results in a scheme with $O(\lambda N^2)$ communication between the client and servers and can be further reduced to $O(\lambda N)$ communication by having users only gossip signatures over the messages they receive instead of the messages themselves. 

The scheme above can provide some of the properties we want from a privacy-preserving payment splitting scheme, but not all of them. In particular, all transactions in that approach are visible to every member of the group, so it does not achieve debtor privacy. In general, it is easy to provide one or the other of our privacy requirements -- debtor privacy without server privacy could easily be achieved in an existing payment splitting app modified to hide some transaction information from users -- but combining them in the same scheme requires more work.

\section{Core Functionality}\label{construction}

This section presents the core functionality for privacy-preserving payment splitting. We begin with a simplified version of our protocol that provides neither efficiency nor security and add in features to improve security and performance one at a time. We will present two variants of our system: one that is only secure against a semihonest server who adheres to the rules of the protocol while trying to learn user secrets and another that is secure against a fully malicious server who can arbitrarily deviate from the protocol. Proofs of security appear in Section~\ref{security}, and we discuss a number of extensions to the core functionality in Section~\ref{extensions}. 

Our constructions derive their security from the assumption that there exists a pseudorandom function (PRF), e.g. that AES is a secure PRF. Informally, a PRF has the property that an adversary cannot distinguish between a random string and the output of a PRF on a given input. The formal definition of a PRF follows.

\begin{definition}[Pseudorandom Functions~\cite{GGM84}] Let $F: \{0, 1\}^n\times\{0,1\}^n\rightarrow\{0, 1\}^n$ be an efficiently computable, length-preserving keyed function. We say that $F$ is a pseudorandom function (PRF) if for all probabilistic polynomial time distinguishers $D$,
\begin{align*}|\textsf{Pr}[D^{F_k}(1^n)=1]-\textsf{Pr}[D^{f_n}(1^n)=1]|\end{align*}
is negligible where $k\gets\{0,1\}^n$ is chosen uniformly at random and $f_n$ is chosen uniformly at random from the set of functions mapping $n$-bit strings to $n$-bit strings. 
\end{definition}

\vspace{.3em}\noindent\textbf{Setup and sharing keys}. A group begins when its $N$ members agree on a shared key. The details of how group members share a key are covered by prior work and constitute an orthogonal problem. One possibility is for one group member to pick the key and send it to the others, encrypting it with a public key for each other member. The server can keep a bank of users' public keys, as is common in encrypted messaging apps, or use a system such as CONIKS\cite{coniks} in the malicious server setting. 

\subsection{The Basic Protocol}\label{overviewSec}

%\vspace{.3em}\noindent\textbf{Correctness and Privacy}. 

\begin{figure}
\centering
\includegraphics[width=.9\linewidth]{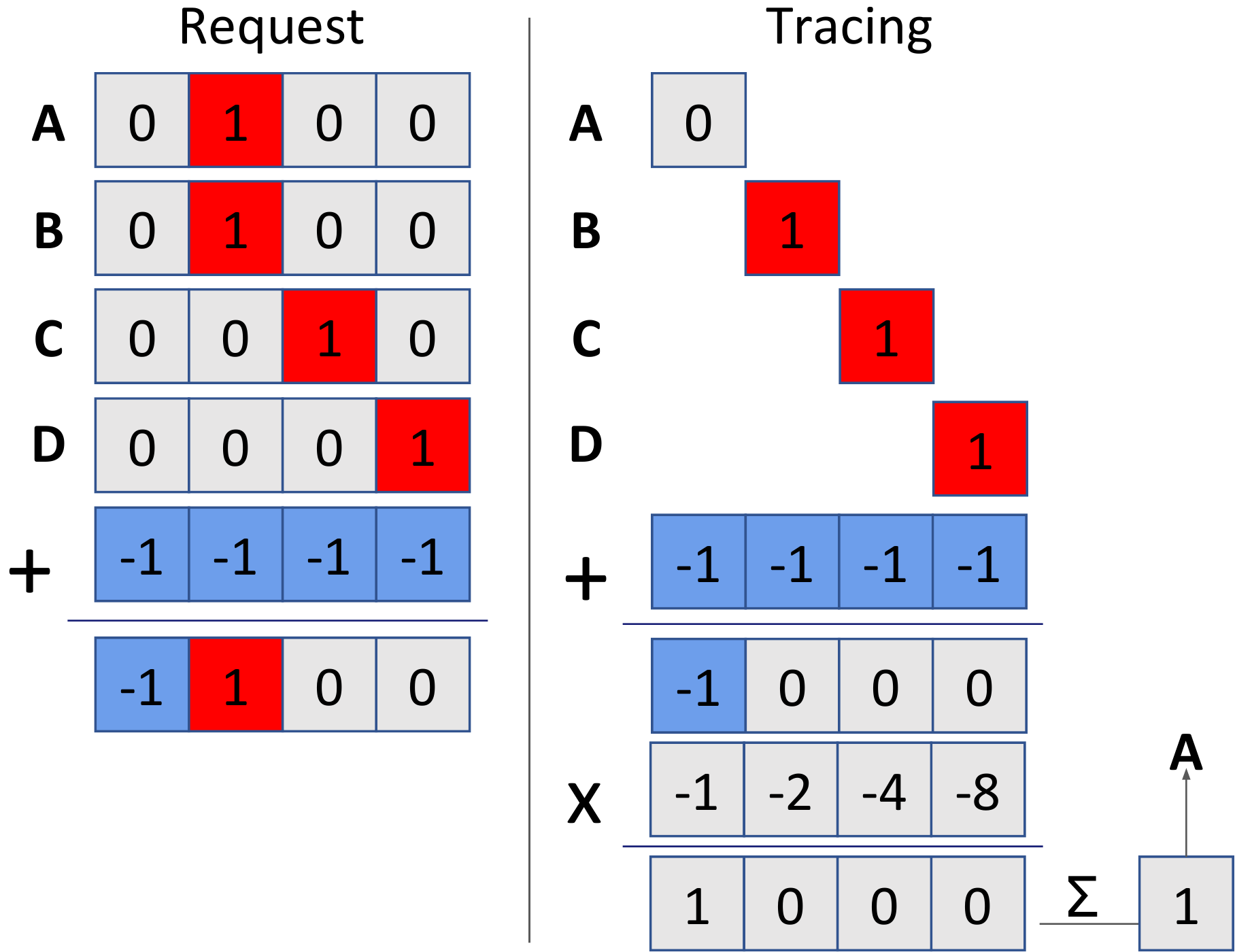}
\caption{Our basic payment splitting protocol. The full protocol uses the output of a PRF whose key is known to all group members to mask and authenticate each user-submitted value. The server also sums all user-submitted values and polls users to ensure the sum matches the group size. Left: an example transaction where user A charges user B one unit of currency. Each user submits a vector with a single nonzero entry, and the server subtracts one from each user's balance, leaving ``-1'' in user A's entry and ``+1'' in user B's entry. Right: the tracing procedure for the same transaction. The server subtracts 1 from what users put in their own entries in the vectors they submit. This results in zero everywhere except for user A, who made a request. The resulting differences are multiplied by increasing powers of 2 and summed to uniquely identify which user(s) made a request in that round. In case of collisions, all users making requests are identified by the sum -- they roll back and resend their requests one by one. }
\label{diagram}
\end{figure}

Over 90\% of our survey respondents reported that transactions in their payment splitting groups are for values under \$100, meaning that large transactions rarely occur. We take advantage of this by designing our protocol to restrict the amount of money that can be transferred in a single transaction and then using the structure resulting from this restriction to enforce integrity requirements without relying on zero knowledge. In the simplified protocol presented here, only one unit of currency changes hands in each round, thus exchanging \$$X$ requires $X$ rounds. Later, we will show how to allow users to exchange \$$X$ using only $\log{X}$ rounds and how to run several rounds in parallel without increasing bandwidth costs. Since payment splitting apps are most frequently used for small quantities of money, these restrictions do not pose a significant performance obstacle for day to day transactions, and an occasional larger transaction can simply be split across multiple rounds. In our final design, a transaction of $\$1,000.00$ requires only 3 rounds to complete. Although we describe our protocol in terms of dollars, we do not require users to round transactions to the nearest dollar, and, in particular, we allow transactions at the 1 cent granularity. 

To introduce the structure of our approach, we will begin by describing our system such that it is missing most privacy and integrity properties. We will then add privacy and integrity one at a time. This basic version of our protocol is demonstrated in Figure~\ref{diagram}.

\noindent\textbf{Setup}. The group begins when its $N$ members agree on a secret key. The server creates a vector of $N$ zeros to represent user balances, with users being assigned numbered $1$ through $N$. %Numbering can be decided when exchanging the group key and does not rely on the trustworthiness of the server. 

\noindent\textbf{Group membership operations}. An existing member adds a new member to the group and sends it the key. The server adds another 0 to the vector of balances and informs the other group members. To leave a group, a user who has paid all his debts and has a balance of 0 notifies the server (who informs other users), and has his entry removed from the vector of balances. A former member who has left a group still holds the group key, so if remaining group members do not trust former members with the key, they must form a new group using a different key. 

\noindent\textbf{Requests}. Payment proceeds in a series of rounds. In each round, a user can request one unit of currency from another. Each user sends the server a vector of size $N$ where all entries are zeros except for a single $1$ in the cell corresponding to the user who will be charged. Users wishing not to make a transaction in a given round put their $1$ in the cell for their own account.

The server sums the vectors it receives from each user and adds that to the balance vector. Then it subtracts 1 from the balance of each user. This results in each user either breaking even or transferring one unit to another user. The values in users' balances in this scheme represent the amount of debt they owe to the group, so putting a 1 in the balance of another user means ``giving'' them one unit of debt. Note that since users only receive updates to their own balance in a given round, the scheme naturally satisfies the notion of debtor privacy. %This process is illustrated in Figure~\ref{architecture}.

\noindent\textbf{Tracing}. In order to support tracing, we first assign each member of the group a power of 2, so group member $i$ is assigned the value $2^i$. Each user then has an additional value associated with them by the server, formed by subtracting 1 from the $i$th element of their request vector. Thus a user who is not making any charges has a 0 and a user who is requesting a payment has a -1. Deriving users' values from the vectors they already sent instead of sending them separately saves one message of bandwidth from each group member to the server. The server multiplies each value with the power of 2 assigned to that group member and sends the sum of the resulting values to each user. 

The users learn from this sum exactly which parties are charging in this round since the resulting value will be unique for each subset of powers of 2. When only one party is charging, a user who disagrees with a charge can reject the charge by initiating a charge of the same value in the opposite direction. Integrity against framing can easily be built on this mechanism by having a user who has been framed report that she has been framed.

In the event that multiple users are charging in the same round, a chargee cannot know which of the chargers originated a particular transaction, so the next few rounds are used to resolve the collision. In the first round, all the transactions of the previous round are ``rolled back'' by having the chargers submit a $+2$ in their own indexes and the chargees a $-1$. In subsequent rounds, each of the charges are made again with the lowest numbered charger making its request in the first round, the second lowest numbered charger making its request in the second round, and so on. Other users wanting to make a charge wait until the resolution process completes to make their charges. This technique can be used to resolve nested collisions too. Since by the time a request completes, the value sent to identify the charger will definitely be a power of 2, there can only be one user who has been charged and the log of that power of 2 will be the identity of the charger.

At first glance, it may appear inconvenient that our system processes transactions within a group one at a time and requires rollbacks and serialization for colliding transactions. However, we found that over 90\% of our survey respondents' groups had no more than 3 transactions in a typical day, meaning that a rollback like the one described here will not need to be used frequently. Given that our protocol supports nested collisions and can operate transparently to the end-user, we see this mechanism as a favorable bandwidth-saving mechanism over a scheme where the server sends much more information to users in each round to avoid rollbacks. 

Note that our user tracing functionality (and user integrity definition) only ensures that the \emph{requestor} in a transaction has not been modified to frame an uninvolved party. This is critical in order for users being charged to reject false charges. 
It is, however, possible for a malicious user to change who \emph{is being charged}. For example, a malicious user who knows that a particular user is going to be charged in a given round can put a 1 in his own index, a -1 in the index of the user who is supposed to be charged, and a 1 in the index of the user to which he wants to attempt to redirect the charge. 
Fortunately, as long as users do not accept charges when they do not owe debts in real life, the fact that the requestor of a charge cannot be tampered with ensures that this kind of framing has limited impact. If the user interface appropriately notifies users about who is charging them (as is common in payment splitting apps used today), they can reject unexpected charges due to this kind of attack just as they would an accidental charge.

\noindent\textbf{Settling}. Users settle by downloading the balance vector from the server and handling payments via another channel.

\subsection{Adding Privacy and Integrity}
We now show how to modify the scheme described above to achieve privacy and integrity.

\vspace{.3em}\noindent\textbf{Privacy}. We provide privacy by masking all the values sent by clients with the outputs of a PRF $f$ evaluated with a key \textsf{sk} chosen during group setup. Let $v_{i,j}$ represent the value in the $j^\textsf{th}$ position of a vector sent by user $i$. For every value $v_{i,j}$ in round $m$, there will be a PRF output $r_{i,j}=f(\textsf{sk}, m||i||j)$ and the value sent to the server instead of $v_{i,j}$ will be $v_{i,j}+r_{i,j}$. %Note that we could equivalently use a PRG  instead of a PRF (where values of $r_{i,j}$ are outputs of the PRG in order) but opt for a PRF for ease of exposition, especially because we will eventually instantiate our PRF with AES.

In order to maintain correctness, each user will compute all values of $r_{i,j}$ in each round and keep a running sum of all the $r$s that have been added into their own balances. They can do this because the PRF key is shared by all members of the group. When receiving balances from the server, users will subtract off the sum of all the $r$ values that have been added to the balance to retrieve their actual balance. When receiving the sum used for tracing, users must subtract off the relevant $r$ values multiplied by the corresponding powers of 2 with which they will have been multiplied by the server. 

\vspace{.3em}\noindent\textbf{Integrity}. After receiving inputs from the clients, the server takes the sum of the values it wishes to add to the balance vector and polls each user to make sure the value is equal to the number of members in the group plus the sum of all $r$ values used as masks in that round. This check protects against silent growth in the total balance of the group. The checks can be done asynchronously with other operations and rounds can be rolled back after-the-fact if an issue is found. Note that this integrity mechanism will work even if the group members are in the middle of rolling back a transaction (as described in the previous subsection) because even though individual users will not send vectors that contain exactly one non-zero entry, the sum of all users' vectors will be the same as in any other round. This mechanism provides the first of our two requirements for integrity, and protection against framing is provided implicitly by our technique for tracing because a user's app can automatically detect if it is traced as the originator of a charge in a round where it did not actually make a request. 

\subsection{Malicious Security}

We can make our construction secure against malicious servers with no changes to the server, no increase in per-round bandwidth, minimal changes on the client, and a one-round increase in the case of colliding transactions. In place of masking each value sent to the server with a pseudorandom $r_{i,j}$, we send the value $sv_{i,j}+r_{i,j}$ for a fixed $s$ also generated at group creation time from the PRF. This value is very similar to the homomorphic MAC tag of Agrawal and Boneh~\cite{AB09} but we use it to achieve different security properties.

We must verify that our scheme still works under this change. The integrity check needs to check that the sum sent from the server equals $s\cdot N$ plus the sum of $r$s instead of just $N$ but otherwise works unmodified. We must also add an integrity check to the setup process because the server cannot generate a vector of 1s on its own and needs to be sent such a vector from one of the users. Other users must check that the vector sent to the server actually consists of all 1s. 

Next, users, upon receiving their new balance from the server, know that the new value will be at most one away from the old value. As such, they can check each of the three possible values by multiplying the possible balances by $s$ and seeing which matches the sent value after removing the $r$ values. If the new value does not match one of the three possibilities, there has either been a collision in payment requests or the server has been caught behaving maliciously. 

The tracing mechanism will require a slightly larger change. First, users do a linear scan of the $N$ possible cases where only one member of the group is making a request by checking if the value received from the server is equal to $-s\cdot 2^i$ for $i\in[N]$ after removing the $r$ values. If none of these match, then there has been a collision, but we are no longer able to determine which users collided. We solve this by using the next round of the protocol to re-send values from this round but without the $s$ multiplied in -- just as we would have in the plain PRF-based construction. After the server sends its responses, the tags from the previous round are used to verify that it did not modify the values and then the identities of the chargers can be checked just as before. The roll-back round must increase the amount by which balances are rolled back to account for the repeat of the colliding transactions being added by the server into balances.

Finally, this construction requires some extra work to settle the group because users will not know what value to expect for each other user's balance. This can be solved by having each user upload a masked value representing their balance and then other users can check to make sure that the value matches the one sent by the server. Since the values sent by a user and the server must match, allowing the users to upload their masked values does not allow users to lie about their balances. 

\subsection{Larger Transactions} 

Sending one unit of currency at a time leads to too many rounds to transact larger amounts. We can modify the round structure so each round is accompanied by a value multiplier. In this version of our scheme, each round has a predetermined value, and all transactions in that round are of that value instead of just 1. The schedule of round values can be fixed at some reasonable configuration (e.g. the first several powers of 2). It is important that the multiplier be applied on the server side (by multiplying each received vector) or else this would open the scheme to attacks on user integrity and break the correctness of our approach for rejecting charges. 

In order to further speed up transactions, several rounds corresponding to different values can take place in parallel. This can be done with little to no increase in bandwidth by using \emph{transaction packing}, where we take advantage of the fact that we instantiate our PRF with AES, which has a 128-bit output. If user balances are unlikely to exceed some large value, say $2^{21}$, we can split each masked value sent to the server into 6 separate ``slots'' where users can put transactions, treating each message as 6 separate transactions, each using 21-bit messages. 
This will result in the server keeping 6 separate 21-bit ``sub-balances'' for each user. The user's total balance is their sum. Summing the sub-balances occurs transparently to the human user whose interaction with the app is unaffected by the optimization.

\subsection{Full Protocol}

Here we formalize the malicious-secure PRF-based construction described above. The semihonest construction is a similar but simpler version of the same protocol and appears in Appendix~\ref{semihonestconstruction}. The construction assumes that users share a key \textsf{sk} as described above and omits details of the key-sharing mechanism. To focus on the core protocol, the construction is written such that each round has value 1 and does not use the transaction-packing optimization described above, but these can easily be added. We say that a party outputs $\bot$ to indicate that an integrity violation has been detected that must be handled out of band, e.g. by users moving to a more trustworthy server or kicking someone out of the group. 

Our fully malicious secure payment splitting scheme $\mathcal{P}_m$ with security parameter $\lambda$ and group size $N$ uses a PRF $f$ with range $\{0,1\}^\lambda$. After setup, the scheme proceeds in rounds, and messages are sent from each client at a rate of one per round. Let $m$ at any time denote the round number at which the current operation began, let $s=f(\textsf{sk},0)$, and let $r_{m,i,j}$ denote $f(\textsf{sk}, \textsf{m||i||j)}$ except $r_{m,i,j}=0$ for values of $m$ smaller than the round in which user $\mathcal{U}_i$ joined the group or after it left. 

\noindent\textbf{Setup}. The server stores a vector \textbf{b} of length $N$ constructed from copies of $0$. User $\mathcal{U}_1$ sends the server a value $a$, which represents a masked version of 1. Let $\textbf{a}$ represent a vector of length $N$ constructed from $N$ copies of $a$. The server sends $a$ to each user, and the users reject if $a\neq s+f(\textsf{sk}, 1)$. The users then store values $b_i=0, b_i'=0$, and the key $\textsf{sk}$. 

\noindent\textbf{Request}. User $i$ requests a unit of currency from user $j$ according to the following steps. 

\begin{enumerate}

\item \emph{Clients prepare vectors}. In the next round $m$, user $\mathcal{U}_i$ creates a vector $\textbf{v}_i$ where $\textbf{v}_{i,j}=s+r_{m,i,j}$ and $\textbf{v}_{i,k}=0+r_{m,i,k} \forall k\neq j$. All other users $\mathcal{U}_k$ create vectors $\textbf{v}_k$ where $\textbf{v}_{k,k}=s+r_{m,k,k}$ and $\textbf{v}_{k,k'}=0+r_{m,k,k'} \forall k'\neq k$. Each user sends its vector $\textbf{v}_i$ or $\textbf{v}_k$ to the server.

\item \emph{Server processes vectors}. Upon receiving the messages from clients for the round, the server takes the sum $\textbf{v}=\Sigma_{i=1}^N\textbf{v}_i$ and also sums the values in \textbf{v} to get $v'$. The server then sets $\textbf{b}=\textbf{b}+\textbf{v}-\textbf{a}$ and computes the value $c=\Sigma_{i=1}^N (\textbf{v}_{i,i}-a)\cdot 2^i$. The server sends the tuple $(v', c, b_l)$ to user $\mathcal{U}_l$, $l\in [N]$. 

\item \emph{Clients check integrity, update balances}. Each user receives the values $(v'^*, c^*, b_l^*)$ from the server. Then there are a number of cases:

\begin{enumerate}
\item \emph{Integrity failure}. If $v'^*\neq sN+\Sigma_{i=0}^N\Sigma_{j=0}^Nr_{m,i,j}$, the user sends an error message to the server who sets $\textbf{b}=\textbf{b}-\textbf{v}+\textbf{a}$ and outputs~$\bot$. 

\item \emph{Balance update or framing failure}. Otherwise, if $c^*=-2^i\cdot s-\Sigma_{j=1}^N((r_{m,j,j}+f(\textsf{sk}, 1))\cdot2^j)$ for some $i\in[N]$, the user checks whether $b_l^*=b_l+xs+\Sigma_{j=1}^Nr_{m,j,l} - f(\textsf{sk}, 1)$ for $x\in\{-1,0,1\}$. If so, it sets $b_l=b_l^*$ and $b'_l=b'_l+x$ for the appropriate $x$. If $v_{i,i}=1$ (if user $i$ did not make a request), user $i$ sends an error message to the server who in turn outputs $\bot$. 

\item \emph{Request collision}. If neither of the above cases apply, then more than one request collided in the same round $m$ (or the server misbehaved). The next two rounds are used to resolve the potential collision. 

\medskip

\textit{Round $m+1$}. In the next round (m+1), the users send the same vectors they sent in round $m$ but with updated values for $r$ and \emph{without} multiplying anything by $s$. Denote the values sent from the server (which behaves the same as above) in round $m+1$ as $(v'^{**} c^{**}, b_l^{**})$. Upon receiving their responses from the server, users check that the values received from the server in this round correspond to the same values received in the previous round (modulo differences due to lack of $s$ and the new values for $r$). 
If any of these checks fail, users output $\bot$.

\medskip

\textit{Round $m+2$}. In the next round (m+2), each party $\mathcal{U}_i, i\in[N]$ submits a vector $\textbf{v}_i$ such that $\textbf{v}_{i,i}=b_i^{**}-b_i-r_{m,i,i}-r_{m+1,i,i}+r_{m+2,i,i}$ and $\textbf{v}_{i,k}=0+r_{m+2,i,k} \forall k \neq i$ and the server behaves as above. Then the various requests that collided in this round are repeated, each in a subsequent round, in order from smallest to largest value of $i_j$ as the requester. 
\end{enumerate}
\end{enumerate}

\noindent\textbf{Trace}. User $j$ checks if anyone has charged her in round $m^*$ as follows. The user checks whether $b_j^*=b_j+xs+\Sigma_{l=1}^Nr_{m^*,l,j}-f(\textsf{sk}, 1)$ for $x\in\{-1,0,1\}$ and, if so, knows she has been charged $x$ units of currency. If no value of $x$ matches, she outputs $\bot$. If her balance has shrunk as a result of the transaction, she looks at the value of $c^*$ in that round and sets $i^*\in[N]$ to be the value for which $c^*=-2^{i^*}\cdot s-\Sigma_{l=1}^N((r_{m^*,l,l}+f(\textsf{sk}, 1))\cdot2^l)$. User $i^*$ is the number of the user who made the charge.

\noindent\textbf{Settle}. The server outputs the vector \textbf{b} to the users. Users respond by sending a fresh masking of their stored values $b'_i$ to the server, who forwards the value to all other users, which in turn unmask and retrieve the value. Users then check that for each $i\in[N]$, $b'_i\cdot s=\textbf{b}_i-m\cdot f(\textsf{sk}, 1)-\Sigma_{m'=1}^m\Sigma_{j=1}^Nr_{m',i,j}$, and reject if any check fails. The output vector $b$ is formed by concatenating the $b_i'$s.

\section{Complexity \& Security}\label{security}

We state the protocol's complexity in terms of a single round, but since our scheme limits the size of transactions, $O(\log X)$ rounds may be needed to send \$$X$ in the general case, resulting in a multiplicative $O(\log X)$ overhead. However, our transaction packing technique can run several rounds in parallel without increasing bandwidth or requiring additional AES evaluations. 

The bandwidth per round for our scheme is $O(N)$ ciphertexts from each user to the server and then $3$ ciphertexts from the server back to each client, each of size $\lambda$. In practice, the ciphertexts in the PRF-based schemes can just be $\lambda=128$ bits long.

Our construction proceeds in rounds such that there is one transaction per round as long as there is no collision between users trying to make charges at the same time. In the case of a collision, the semihonest scheme loses 2 rounds (detect collision, roll back transactions) whereas the malicious secure scheme loses 3 rounds (retrieve non-tag version of values, detect collision, roll back transactions). 

Settling requires one message of size $O(\lambda N)$ from the server to each user in the semihonest case, and in the malicious server case this is preceded by messages of size $O(\lambda)$ from each user to the server. 

\vspace{.3em}\noindent\textbf{Server}.
The server does the same process in each round in both the semihonest and malicious constructions and runs in time $O(\lambda N^2)$. The settle operation requires $O(\lambda N)$ work on the server to send the stored balances, and server storage is $O(\lambda N)$ to hold balances.

\vspace{.3em}\noindent\textbf{Client}.
Our solution requires $O(\lambda N^2)$ time to generate the PRF outputs and take the necessary sums. Achieving malicious security adds lower order terms for various checks but still requires $O(\lambda N^2)$ time in a normal round, with the potential cost of an extra round in case of a collision (see above). Settling can be done in $O(\lambda N)$ time by the client if the large sum that needs to be taken is built incrementally and saved during each round. Although $O(\lambda N)$ computation for each round would be more desirable, we find in our evaluation that performance for realistic group sizes remains quite fast. 

\vspace{.3em}\noindent\textbf{Security}. We now state our security theorems for both the semihonest and fully malicious constructions. We defer proofs to Appendix~\ref{proofsAppendix}, which proves that our scheme satisfies the properties of a secure payment splitting scheme as described in Section~\ref{secmodel}. 

\begin{theorem}\label{semihonest}
Assuming $f$ is a secure PRF, the semihonest secure payment splitting scheme $\mathcal{P}_p$ has correctness, server privacy, debtor privacy, and user integrity. 
\end{theorem}

\begin{theorem}\label{malicious}
Assuming $f$ is a secure PRF, the fully malicious secure payment splitting scheme $\mathcal{P}_m$ has correctness, server privacy, debtor privacy, user integrity, and server integrity (against a malicious server). 
\end{theorem}

\section{Extensions}\label{extensions}
This section briefly describes a number of extensions to the core protocol that may be useful in practice. 

%\subsection{Identifying Misbehaving Parties}

\vspace{.3em}\noindent\textbf{Identifying misbehaving users}. Our user integrity checks, as described thus far, allow users to detect whether a user has misbehaved, but we would also like to be able to determine \emph{which} users have misbehaved in order to punish them or prevent membership in future groups. We can easily accomplish this by having the server send each user all the messages it received in the previous round, so users can check to see whose input was malformed. Unfortunately, this approach requires the server to send each user $N^2$ ciphertexts and, more importantly, compromises debtor privacy by revealing who was charged in that round. It is, however, possible to identify misbehaving users without breaking privacy. User integrity requires that the sum of all values in all users' vectors is $1$. Observe that in order to tell whether a user violated integrity, other users only need to learn whether the user submitted a vector whose entries do not sum to $1$. As such, the server only needs to distribute the sum of each user's vector, and any user whose vector does not sum to $1$ must be misbehaving (except in the case of a transaction that rolls back a collision, but all users will be aware that this is happening). This does not compromise debtor privacy because the sum for an honest user will always be $1$ regardless of whether or not that user is involved in a transaction. Moreover, the server now only needs to send each user $N$ ciphertexts instead of $N^2$.

\vspace{.3em}\noindent\textbf{Handling framing}. User integrity also requires that a user who has not initiated a charge can detect that she is being framed. An additional practical concern is for the victim of framing to prove his or her innocence to other users and to identify the misbehaving user who did the framing. Our basic scheme does not offer a mechanism for other users to verify a claim that someone has been framed. We can, however, add such a mechanism without much work. A user can prove innocence when framed by asking the server to send every user the single entry in her vector corresponding to the index of the user she has purportedly charged. If that entry is a zero (as it will always be if she did not really make a charge), then she has clearly been framed. We stress that it is not the human user herself who detects and proves that she has been framed, but the user's app, which sees it has been traced as the requester in a charge yet knows that the human user has not charged anyone in that round. The problem of detecting who has done the framing without weakening debtor privacy appears more difficult, and we leave this problem for future work. 

%The solution below doesn't work because in unmasking the 1s and 0s, the users will be able to unshuffle the sums they were sent. :(

%Recall that an honestly-formed vector in our scheme consists of all zeros except a $1$ in a single position. We will verify this invariant for each user without revealing anything about the position of the $1$ in that user's vector. To investigate a user for framing, the server distributes a series of sums $s_1, ..., s_N$, defined as $s_i\gets\Sigma_{j\in[N]\setminus i}v_j$, where $v_1, ..., v_N$ are the entries in the vector submitted by that user. Importantly, the sums are sent out in a random order. To verify $s_1, ..., s_N$, users must check that every value $s_i$ corresponds to $1$, except one value which corresponds to $0$ (the sum where the 1 in the vector was the excluded entry). Any maliciously formed vector will be caught by this check because there will be more than one non-$1$ entry, but every honest vector will result in the same series of randomly shuffled sums, regardless of whether that vector represented a charge or not. 

%\subsection{Handling Users Going Offline}

\vspace{.3em}\noindent\textbf{Handling users going offline}. Our schemes can be modified to handle users going offline. Intuitively, the server supplies a vector that represents making no charges on behalf of any user who does not send anything in a round. The server adds in a plaintext vector for any missing users and notifies others about who did not send a message in a given round so that they know whose $r$ values to omit from the various sums. 

The same approach works in the malicious security solution, but here we need to make a tradeoff. Allowing the server to be resilient to offline users means giving it the power to exclude any users' transactions by saying that they were not present for that round. Of course, the server does not know what a user is doing in any round, meaning this kind of attack amounts to a denial of service possibility (and the server could always deny service more directly), but our security definition would need to be modified to explicitly allow this kind of omission on the part of the server. We do note that this denial of service attack would be \emph{undetectable} by other users, a consideration which should be taken into account before deciding whether or not to enable this feature. 

Our security guarantees degrade gracefully in the absence of some users, with a transaction's anonymity set always being equal to the number of online users and user integrity holding so long as one honest user is online. The definitions in Section~\ref{secmodel} all describe a setting where users are always online, providing an anonymity set of size $N$ for each transaction and enforcing user integrity against $N-1$ malicious users. When the number of online users is reduced to $N'<N$, our scheme provides an anonymity set of size $N'$ and user integrity against $N'-1$ users. Even if no honest users are online when a malicious transaction is made, they can still detect a malicious transaction when they come back online later and the server sends them messages they missed. 

%\subsection{Improving Usability for Tracing and Charge Requests}

\vspace{.3em}\noindent\textbf{Improving usability for tracing and charge requests}. A user making a charge can indicate the total amount they wish to charge (split across several rounds) by sending a second encrypted value containing the amount whereas other users upload an encryption of zero. The server sums the encrypted values and sends the result to each client. This way clients know how much they will be charged at once and can give the app permission to accept charges for the appropriate number of rounds. This would enable a UI not so different from that used in payment splitting apps today while the app handles the details of the underlying protocol. At the same time, this does not introduce new security concerns because additional charges beyond the amount claimed will register as new transactions, and users' real balances are not affected by this bookkeeping shortcut. 

%\subsection{Payment Splitting with Collateral}
\vspace{.3em}\noindent\textbf{Payment splitting with collateral}. Payment-splitting groups sometimes encounter a problem where one member incurs debt to others and repeatedly fails to pay. One solution to this problem involves users putting up money as collateral when joining a group in order to insure their debts with a deposit should they prove untrustworthy. We can make our scheme compatible with this remedy as well. The core idea is simple. At regular billing intervals, say monthly, users provide the service provider with an additively homomorphic Pedersen commitment~\cite{Ped91} to their current balance and a proof that the commitment is to a value less than their deposit. If they cannot produce such a proof, they must deposit more money until they reach a point where they can. 

The challenge lies in producing a commitment to a value that provably corresponds to a user's balance. This can be accomplished by soliciting the assistance of other users, each of which can submit to the server a commitment to the masking value hiding the target user's balance. Since all users share a PRF key, they can use the PRF to generate randomness for the commitment, resulting in every member sending the same commitment to the same value. Thus, so long as one user is honest, the server will not be deceived as to the masking value. Next, the server can, on its own, create a commitment to the target user's masked balance, which it already knows. Once the server has a commitment to the masked balance and the mask value, it can subtract the mask and get a commitment to the user's actual balance, which the user can then (in zero-knowledge) prove is less than the deposit. 

\begin{figure}
\centering
\includegraphics[width=.60\linewidth]{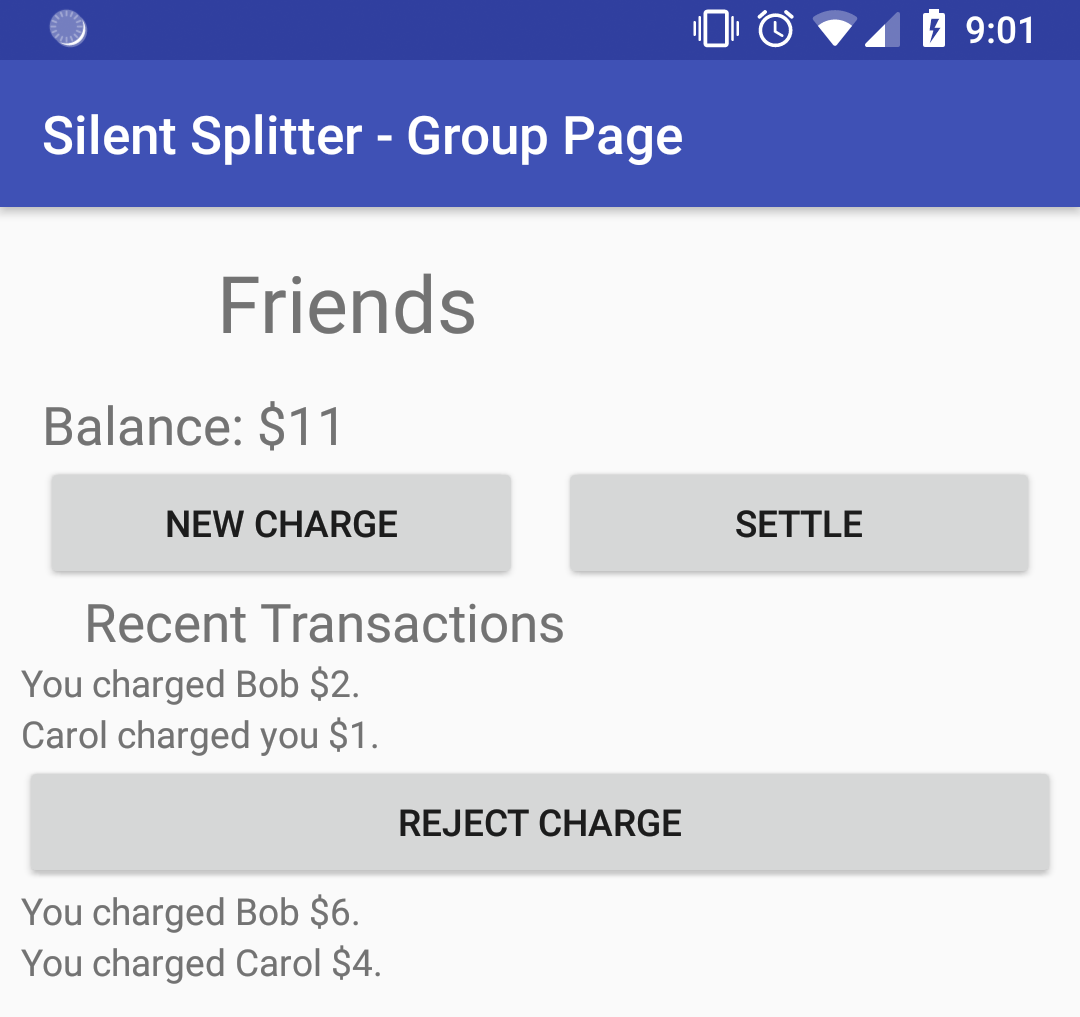}
\caption{A screenshot of our Android app. The user has sent charges to Bob and Carol and also been charged $\$1$ by Carol, which can be rejected via the ``Reject Charge'' button. }
\label{screenshot}
\end{figure}

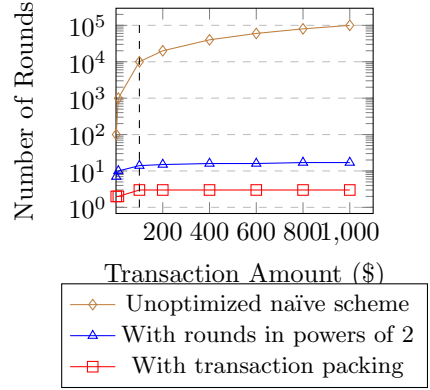
\begin{figure}
\centering
\begin{tikzpicture}
\centering
\begin{semilogyaxis}[
log origin = infty,
title={Rounds Needed to Process Transactions},
ylabel = {Number of Rounds},
width=5cm,
xlabel={Transaction Amount (\$)},
xtick={200, 400, 600, 800, 1000},
ytick={1, 10, 100, 1000, 10000, 100000},
%legend pos=north west,
legend style={at={(0.5,-0.34)},
 legend entries={\small Unoptimized na\"ive scheme, \small With rounds in powers of 2, \small With transaction packing}, anchor=north, legend columns=1},
xmin=0,
xmax=1100,
ymajorgrids=true,
grid style=dashed
]
\addplot[color=brown, mark=diamond]
coordinates{
  (1, 100)
  (10, 1000)
  (100, 10000)
  (200, 20000)
  (400, 40000)
  (600, 60000)
  (800, 80000)
  (1000, 100000)
};

\addplot[color=blue, mark=triangle]
coordinates{
  (1, 7)
  (10, 10)
  (100, 14)
  (200, 15)
  (400, 16)
  (600, 16)
  (800, 17)
  (1000, 17)
};

\addplot[color=red, mark=square]
coordinates{
  (1, 2)
  (10, 2)
  (100, 3)
  (200, 3)
  (400, 3)
  (600, 3)
  (800, 3)
  (1000, 3)
};

%\legend{\small Plain scheme, \small Rounds in powers of 2, \small Rounds in powers of two AND transaction packing}

\coordinate (A) at (axis cs: 100,1);
\coordinate (B) at (axis cs: 100,150000);
\draw [black,sharp plot,dashed] (A) -- (B);
\end{semilogyaxis}
\end{tikzpicture}
\caption{Number of rounds needed for transactions of different amounts as we add our bandwidth-saving optimizations. Over 90\% of users in our survey reported that their typical transactions fell to the left of the vertical dotted line, and even a $\$1,000$ transaction requires only 3 rounds, a $3,300\times$ improvement over the na\"ive scheme. }
\label{roundsgraph}
\end{figure}

\section{Implementation and Evaluation}\label{eval}
We implemented our system as an Android app and an accompanying server application using the Java Spark framework~\cite{javaspark}. Our app allows users to create and join groups, send charges to each other, reject unwanted charges from other users, and settle the group balance if desired. We did not implement the optimizations for supporting large transactions and rollbacks. Figure~\ref{screenshot} shows a screenshot of a group's page in the app. 

We implemented our app as a simple front-end for the functionality provided by our protocol, but there is no technical limitation preventing more complex user interfaces -- such as those available in non-private payment splitting apps in use today -- being placed in front of our protocol. For example, our app allows users to individually charge other users in a group. A more complex app could provide an interface where the user enters a charge that is assigned to the whole group or some subset of the group, e.g. $\$80.56$ split among $4$ people, and the app could split this into three charges of $\$20.14$ each for other diners to repay whoever paid the bill. Observe that such a splitting would not run into issues with colliding charges because the same user is initiating all the charges and can stagger them over several rounds. This way charges from the same real-world transaction always take place one at a time. 

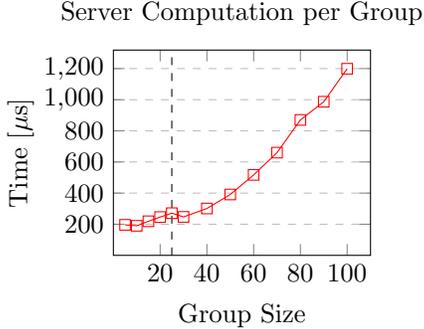
\begin{figure}
\centering
\begin{tikzpicture}
\centering
\begin{axis}[
title={Server Computation per Group},
ylabel = {Time [$\mu$s]},
width=5cm,
xlabel={Group Size},
xtick={20, 40, 60, 80, 100},
ytick={200, 400, 600, 800, 1000, 1200},
xmin=0,
ymin=0,
ymajorgrids=true,
grid style=dashed
]
\addplot[color=red, mark=square]
coordinates{
  (5, 196)
  (10, 190)
  (15, 219)
  (20, 246)
  (25, 271)
  (30, 247)
  (40, 301)
  (50, 391)
  (60, 517)
  (70, 660)
  (80, 870)
  (90, 987)
  (100, 1199)
};

\coordinate (A) at (axis cs: 25,0);
\coordinate (B) at (axis cs: 25,1300);
\draw [black,sharp plot,dashed] (A) -- (B);
\end{axis}

\end{tikzpicture}
\caption{Server round computation time (in microseconds) for one group. The same computation applies for both the semihonest and malicious server settings. 92\% of users in our survey reported that their largest group size fell to the left of the vertical dotted line ($271\mu s$).}
\label{servergraph}
\end{figure}

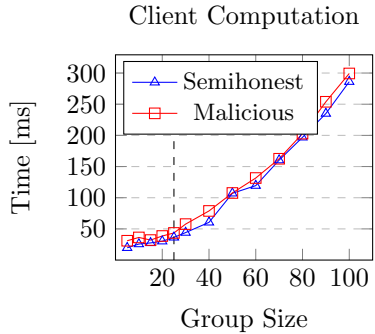
\begin{figure}
\centering
\begin{tikzpicture}
\centering
\begin{axis}[
title={Client Computation},
ylabel = {Time [ms]},
width=5cm,
xlabel={Group Size},
xtick={20, 40, 60, 80, 100},
ytick={50, 100, 150, 200, 250, 300},
legend pos=north west,
xmin=0,
ymin=0,
ymajorgrids=true,
grid style=dashed
]
\addplot[color=blue, mark=triangle]
coordinates{
  (5, 19.7)
  (10, 25.7)
  (15, 27.5)
  (20, 30.2)
  (25, 36.4)
  (30, 44.1)
  (40, 60.5)
  (50, 107.1)
  (60, 119.4)
  (70, 160.3)
  (80, 196.9)
  (90, 235.0)
  (100, 286.1)
};

\addplot[color=red, mark=square]
coordinates{
  (5, 31.0)
  (10, 36.3)
  (15, 32.1)
  (20, 38.6)
  (25, 43.5)
  (30, 57.7)
  (40, 78.8)
  (50, 107.6)
  (60, 131.6)
  (70, 162.5)
  (80, 202.8)
  (90, 253.9)
  (100, 299.3)
};

\legend{\small Semihonest, \small Malicious}

\coordinate (A) at (axis cs: 25,0);
\coordinate (B) at (axis cs: 25,400);
\draw [black,sharp plot,dashed] (A) -- (B);
\end{axis}
\end{tikzpicture}
\caption{Client round computation time (in milliseconds) for the semihonest and malicious server settings. 92\% of users in our survey reported that their largest group size fell to the left of the vertical dotted line ($36.4/43.5$ms).}
\label{clientgraph}
\end{figure}

We evaluated the performance of our implementation on a MacBook Pro with a 2.9 GHz Intel Core i7 processor and 16GB RAM running macOS Sierra Version 10.12.6 for the server and a Google Pixel (1.6GHz quad-core processor, 4GB RAM) running Android 8.1.0 for the app. Our performance tests measure the server and client running time for one round, using AES128 (default Android implementation) as our PRF. The reported times include all computation outside of network communication and are averages over 10 rounds each on the server side and 5 each on the app. We measured performance for groups of up to size 100 because no users reported having a group of size greater than 100 in our survey. Most groups were reported to be of less than 10 members (69\%) or less than 25 members (92\%), a setting where our protocol performs particularly well.

\vspace{.3em}\noindent\textbf{Rounds}. Since we operate and measure our system in terms of rounds, it is important to understand how many rounds would be required to make realistic payments. Figure~\ref{roundsgraph} shows the number of rounds required to process transactions of various sizes and the impact of our optimizations in reducing the number of rounds. Our optimizations reduce the number of rounds required from linear to logarithmic in the transaction value, resulting in savings of $500\times$ and $3,300\times$ for $\$10$ and $\$100$ transactions respectively. A $\$100$ transaction requires 2 rounds to complete, and transactions of over $\$1,000$ still require only $3$ rounds. Recall that we support transactions at the granularity of 1 cent, so a $\$1,000$ transaction corresponds to the transfer of $100,000$ cents. Returning to the example from the beginning of this section, three $\$20.14$ charges sent to participants in a dinner that cost $\$80.56$ would run in two rounds each, resulting in six rounds to complete the transaction.

\vspace{.3em}\noindent\textbf{Server Performance}. Figure~\ref{servergraph} shows the per-round server side running time for a single group, which ranges from about 200 \emph{micro}seconds for a group of 10 members to about 1.2 milliseconds for a group of size 100. Since there are no server side changes to the protocol between the semihonest and malicious server settings, performance in the two regimes is identical. Computation for each group operates entirely independently of other groups, so a server can scale perfectly to a larger number of cores, enabling a server only as powerful as a commodity laptop to handle several thousands of groups per second. The extreme efficiency and scalability of our scheme results from the fact that the server does not execute any cryptographic operations, only addition on 128-bit integers. 
Server memory requirements per group are also small because the server can add users' inputs into running totals for each round as they arrive, removing the need to keep all messages for a given round in memory until it completes. 

\vspace{.3em}\noindent\textbf{Client Performance}. Figure~\ref{clientgraph} shows the running time for the Android app to compute its inputs and process the outputs for each round. Running times range from 26-286ms in the semihonest setting and 36-299ms in the malicious setting. The overhead of malicious server security over semihonest security is quite small in both relative and absolute terms -- less than 5\% for groups of 100 members and never more than 20ms. Recall that the difference between semihonest and malicious security here applies only to the server security protections, and both protocols provide security against malicious users. The lightweight nature of this computation, consisting primarily of PRF evaluations and 128-bit additions, means it can conveniently run as a regular background process without causing an undue burden on a user's phone. 

\vspace{.3em}\noindent\textbf{Bandwidth}. The bandwidth of each round is $16N$ bytes, where $N$ is the group size, from each user, and $52$~bytes---three 16-byte values and one 4-byte status code---sent back from the server to each user regardless of group size (plus a negligible constant for sending the data in JSON format). For group sizes used in practice, this does not prove to be prohibitively large. A user in a group of size 100 would only send about 1.6KB of data, and users in the more commonly reported group sizes of 10 or 25 would send 160B or 400B respectively. 

\begin{figure}
\centering
\begin{tikzpicture}
\centering
\begin{axis}[
title={Computation Costs},
ylabel = {Time [ms]},
width=5cm,
xlabel={Number of Rounds},
xtick={4, 8, 12, 16, 20, 24},
ytick={200, 400, 600, 800, 1000, 1200},
legend pos=north west,
xmin=0,
xmax=25,
ymin=0,
ymax=1400,
ymajorgrids=true,
grid style=dashed
]
\addplot[color=blue, mark=triangle]
coordinates{
  (2, 800)
  (4, 800)
  (6, 800)
  (8, 800)
  (10, 800)
  (12, 800)
  (14, 800)
  (16, 800)
  (18, 800)
  (20, 800)
  (22, 800)
  (24, 800)
};

\addplot[color=red, mark=square]
coordinates{
  (2, 73)
  (4, 146)
  (6, 219)
  (8, 292)
  (10, 365)
  (12, 438)
  (14, 511)
  (16, 584)
  (18, 657)
  (20, 730)
  (22, 803)
  (24, 876)
};

\legend{\small ZKLedger, \small Our System}

\coordinate (A) at (axis cs: 3,0);
\coordinate (B) at (axis cs: 3,1200);
\draw [black,sharp plot,dashed] (A) -- (B);
\coordinate (C) at (axis cs: 6,0);
\coordinate (D) at (axis cs: 6,1200);
\draw [black,sharp plot,dashed] (C) -- (D);
\end{axis}
\end{tikzpicture}
\\\vspace{1em}
\begin{tikzpicture}
\centering
\begin{axis}[
title={Communication Costs},
ylabel = {Communication~[kB]},
width=5cm,
xlabel={Number of Rounds},
xtick={4, 8, 12, 16, 20, 24},
ytick={10, 20, 30, 40, 50, 60, 70},
legend pos=north west,
xmin=0,
xmax=25,
ymin=0,
ymax=75,
ymajorgrids=true,
grid style=dashed
]
\addplot[color=blue, mark=triangle]
coordinates{
  (2, 45)
  (4, 45)
  (6, 45)
  (8, 45)
  (10, 45)
  (12, 45)
  (14, 45)
  (16, 45)
  (18, 45)
  (20, 45)
  (22, 45)
  (24, 45)
};

\addplot[color=red, mark=square]
coordinates{
  (2, 4.4)
  (4, 8.8)
  (6, 13.2)
  (8, 17.6)
  (10, 22)
  (12, 26.4)
  (14, 30.8)
  (16, 35.2)
  (18, 39.6)
  (20, 44)
  (22, 48.4)
  (24, 52.8)
};

\legend{\small ZKLedger, \small Our System}

\coordinate (A) at (axis cs: 3,0);
\coordinate (B) at (axis cs: 3,1200);
\draw [black,sharp plot,dashed] (A) -- (B);
\coordinate (C) at (axis cs: 6,0);
\coordinate (D) at (axis cs: 6,1200);
\draw [black,sharp plot,dashed] (C) -- (D);
\end{axis}
\end{tikzpicture}
\caption{Computation and communication cost comparison to ZKLedger~\cite{zkledger}. The vertical lines represent the number of rounds needed for a single \$1,000 transaction (in 3 rounds) and splitting a \$80.56 dinner among 4 people (in 6 rounds).}
\label{ZKcomparison}
\end{figure}
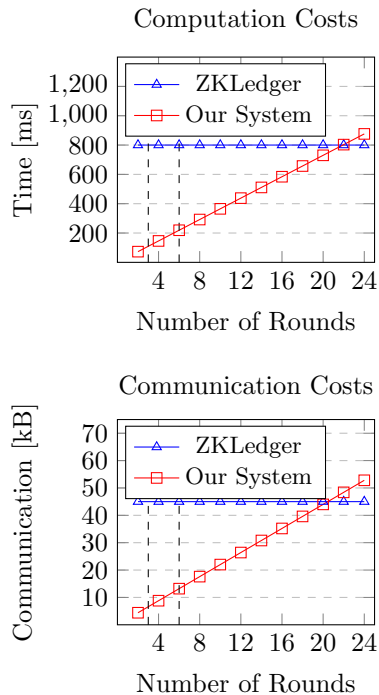

\ignore{
\begin{table}
   \caption{Performance of our system compared to payment splitting with ZKLedger for a group of size 10.}
   \label{tab:perfvszkl}
   \centering

   \begin{tabular}{lrr}
                                    & ZKLedger & Our system \\
      \midrule
      Bandwidth per \$1,000 transaction & 45~kB    & 6.6~kB \\
      Compute time per \$1,000 transaction      & 800~ms   & 110~ms \\
       Bandwidth to split \$80.56 dinner for 4 & 45~kB    & 13.2~kB \\
      Compute to split \$80.56 dinner for 4      & 800~ms   & 220~ms \\
   \end{tabular}
\end{table}
}

\vspace{.3em}\noindent\textbf{Comparison to ZKLedger}. ZKLedger~\cite{zkledger} allows parties with access to a distributed ledger to privately record transactions and ensure the public integrity and auditability of the ledger, primarily targeting large financial institutions. Somewhat similarly, our scheme aims to allow groups of private individuals to record debts without compromising transaction privacy or integrity.  It should be noted that there are important differences between the two settings. Our scheme runs in a continuous series of rounds to hide who initiates transactions whereas entries are only added to ZKLedger when a transaction takes place. ZKLedger offers additional auditing functionalities that are not relevant to our use case. On the other hand, ZKLedger does not offer an in-protocol mechanism for users to contest charges, whereas our protocol does. Finally, ZKLedger's transactions must go on an ever-growing ledger to enable external auditing whereas we have no such requirement. Both systems, however, could potentially be used for payment splitting applications. 

To compare fairly with ZKLedger, we modify our scheme so a round only takes place when a user initiates a transaction. This  leaks the identity of the user initiating a charge, similar to ZKLedger. The change does not affect other security properties.

ZKLedger's code is not public, so we can at best compare to performance numbers reported by their paper, which were taken using virtual machines, each with 4 cores of Intel Xeon E5-2640 2.5 GHz processors, 24GB of RAM, and running 64-bit Linux 4.4.0 on Ubuntu 16.04.3. Our system outperforms ZKLedger in the payment splitting application despite running most of its computation on a mobile phone processor. 

We calculate our system's computation time as the sum of client and server running times for each round. Figure~\ref{ZKcomparison} compares our performance to ZKLedger for a ten member group (the largest for which ZKLedger reports complete transaction times).
Our scheme processes a $\$1,000$ transaction with 7.3$\times$ less computation and with $6.8\times$ lower bandwidth than ZKLedger, and it splits a $\$80.56$ dinner bill among 4 diners with 3.7$\times$ less computation and with $3.4\times$ lower bandwidth. The tipping point where ZKLedger outperforms our system is when a single transaction requires over 22 rounds to complete. 
Our faster performance comes from using only AES in each round (ZKLedger uses Pedersen commitments~\cite{Ped91}) and the absence of large zero-knowledge proofs from our design. 

Note that if we were to use our unmodified system (where the initiator of each charge is hidden from the server and rounds occur at fixed time intervals) to make each transaction above, the transaction latency would be determined by the time taken per round, a configurable parameter. As such, we would expect such a transaction to clear more slowly if we also wish to hide who initiates transactions. 

\section{Future Work}\label{discussion}
This section covers modifications to our architecture and security model that could be explored in future work. 

%\subsection{Integration with Existing Payment Systems}
%cut for space
\ignore{
\vspace{.3em}\noindent\textbf{Integrating with existing payment systems}. 
%One benefit of our solution is that we allow different levels of privacy and functionality tradeoffs in the app running on top of it. 
Our solution can be deployed in a setting where each user signs up with a pseudonym for each group and handles payments outside the app or, at the other end of the spectrum, each user can have a public identity affiliated with many groups and even upload payment information for server assistance in the process of settling at the end of the group's lifespan. 
In the latter setting, the users would send the server their balances at the end of the group's lifespan to handle settling with an external service. 
Since user balances must sum to zero, the server knows if a user has lied about his balance. 
The server would necessarily learn how much each person needs to pay or receive, but the details of payers, payees, quantities, and frequencies of past transactions would remain hidden from the server. 
This point on the tradeoff between privacy and convenience gives users increased privacy while preserving the flexibility of choosing any money-transfer mechanism. 
}
\vspace{.3em}\noindent\textbf{Alternative system architectures}. We built our system to conform with the prevailing architecture of payment splitting apps in use today, where clients connect to a central server which provides the payment splitting service. 
Although this architecture is particularly interesting due to its relevance to practice, a number of other possibilities remain unexplored. For example, the techniques used for query compression in  Riposte~\cite{dpfs, riposte} could be directly applied to a multiserver port of our system, reducing per-round bandwidth per user to square root or even logarithmic~\cite{logdpfs} in the group size. Settings where group members communicate among themselves may also lead to more efficient schemes. 

\vspace{.3em}\noindent\textbf{Alternative applications}. In addition to considering other architectures for solving the payment splitting problem, we may also consider other problems that may be solved by our current system architecture. Although we have focused on payment splitting, our solution can apply to many situations where a group needs to privately keep records. For example, our system could be used to allow a third-party IOT device manufacturer to provide private analytics software. Devices send ``payment requests'' when they use energy or when a user interacts with them, with a separate balance allocated for each aggregate being measured by the software. The group of devices could ``settle'' at the end of a month and reveal aggregated analytics information to the manufacturer without revealing the details of when a given device was used. Most generally, our solution to payment splitting can be seen as an approach to metadata-hiding communication that sends particularly structured messages.

%\subsection{Group Hiding}
\vspace{.3em}\noindent\textbf{Group hiding}. 
%Our scheme operates at the level of a group and hides all information about transactions within the group, but leaks the membership of the group to the service provider. 
A natural extension of our scheme would be to hide group membership from the server. We observe that hiding group membership fundamentally changes the parameters of the problem. 
In our setting, the server can treat each group separately, allowing the number of groups to scale rapidly for realistic group sizes. 
In the group hiding setting, however, since the server cannot determine which users belong to the same group, server behavior must be independent of group constitution. At this point, it may be more effective to dispense with the structure of groups and focus on the more general problem of processing confidential transactions between arbitrary users. Our work solves the real-world problem of privacy in payment splitting groups, but our techniques do not appear to extend directly to handling general confidential transactions. 

\section{Related Work}\label{related}

A number of generic cryptographic tools could be used to achieve a functionality similar to ours. Most directly, fully homomorphic encryption (FHE)~\cite{Gentry09,BGV11,GSW13} could be used for clients with a shared key to outsource any computation to an untrusted server. Although payment splitting is a special case of what can be accomplished with FHE, FHE remains impractical for most use cases today. Moreover, even with FHE, we would require additional safeguards for integrity against a malicious server.

Somewhat more practical are multiparty computation (MPC) techniques. Although the typical setting for MPC~\cite{GMW87} where multiple parties interact with each other to compute a function does not apply to our setting, there are works that focus on MPC between users with the assistance of a single server~\cite{FKN94,KMR11,KMR12,HLP11}. Kamara et al.~\cite{KMR11,KMR12} extend an earlier garbled-circuits~\cite{Yao82b,Yao86} based approach of Feige et al.~\cite{FKN94} to such a setting, but since they work with garbled boolean circuits, their approach would incur a sizeable overhead in keeping track of user balances compared to ours. 

%cutting to save space
\ignore{One could view privacy-preserving payment splitting as a variant of metadata-hiding messaging.
Recent works have taken strides toward making such systems practical~\cite{riposte,vuvuzela,stadium,pung,atom}, but they often operate in settings with multiple servers or provide weaker security guarantees. Even an ideal metadata-hiding anonymous messaging scheme, however, may fall short of efficiently satisfying our user integrity and debtor privacy requirements, as it is not immediately clear how to efficiently police users who send different messages to different members of their group. The important distinction between messaging and payment splitting is that messaging does not impose any \emph{semantic} requirements on the content of data, whereas payment splitting does. }

To our knowledge, our work is the first to directly consider privacy in payment splitting. However, many works deal with related problems in the space of payments and privacy. The notion of an anonymous digital currency in a centralized setting was first proposed by Chaum~\cite{CHaum82,Chaum83} and has been the subject of almost continuous study since~\cite{CFN88,OO89,Brands93,Camenisch98,CFT98,CHL05,CHL06,BCFK15,WKCC18}. 
Since the advent of Bitcoin~\cite{bitcoin}, decentralized cryptocurrencies have come to the forefront of research on privacy and digital payments. While Bitcoin itself only provides pseuodonymity to its users and has been shown to admit tracing of user identities~\cite{AKR+13,MPJ+13}, a number of proposals for modifications or alternative cryptocurrencies provide stronger privacy guarantees~\cite{mw,conftrans,zerocash}. ZKLedger~\cite{zkledger} has a similar flavor to our work in terms of private auditing but targets a very different setting and at higher cost, see Section~\ref{eval}. %cite zether later?
By focusing on the special case of payment splitting, we build more efficient solutions than are possible in the general case of anonymous payments.

Our solution for malicious security is similar to the Homomorphic MACs of Agrawal and Boneh~\cite{AB09}, which belong to a class of works dealing with computation on \emph{authenticated} data first proposed by Johnson et al~\cite{DJM+02}. Our construction can be viewed as a special case of the MACs of Agrawal and Boneh, but the security notions we require do not exactly align with theirs. Their security game allows an adversary to produce a valid tag on any linear combination of previously produced messages, but we require that an adversary can only send the exact results of the server's designated computation and no other function of user-provided inputs. 

\section{Conclusion}\label{conclusion}
We have presented a payment splitting app that hides all transaction data from the service provider. 
We showed how we achieve privacy and integrity in the face of malicious users or a malicious server while only relying on lightweight cryptography on user devices and computing no cryptographic operations whatsoever on the server side. 
Our core protocol operates in rounds, and we showed in our evaluation that it can scale to large numbers of groups, requiring less than 300 microseconds of computation per round for the vast majority of groups. Likewise, the mobile app requires less than 50 milliseconds of computation per round on a user's phone, providing users with improved privacy in a payment splitting app at very little computational cost. 

\section*{Acknowledgment}

We would like to thank the anonymous reviewers for their helpful feedback in improving the paper. 

The majority of this work was completed during a summer internship at Visa Research. In addition, this work was supported by NSF, DARPA, ONR, and the Simons Foundation.

\bibliographystyle{plain}
\bibliography{main} 

\appendix 

\section{Survey Text}\label{survey}

This is a survey about your use of payment splitting apps. This refers to apps which track credit/debit between peers, such as Splitwise, and does NOT include apps that are designed mainly for one-off payments, like Venmo. Since both kinds of apps can make charges to your credit card, a good way to tell them apart is that your balance can only go negative in a payment splitting app.

Some examples of payment splitting apps: Splitwise, Receipt Ninja, BillPin, SpotMe, Conmigo, and Settle Up.

The results of this survey will be used as part of an ongoing Visa Research project on payment splitting apps. Feel free to contact [redacted] with any questions or comments.

\begin{enumerate}
\item What payment splitting app(s) do you use? If you don't use any, write ``none'' and answer the questions below with regard to however you do split payments.

\item How big is the largest group you have in such an app?
\begin{enumerate}
\item $<10$
\item $10-25$
\item $25-50$
\item $50-100$
\item $>100$
\end{enumerate}

\item How much money is used (per person) in one of your typical transactions on this app?
\begin{enumerate}
\item $<$\$$10$
\item \$$10-30$
\item \$$30-50$
\item \$$50-100$
\item $>$\$$100$
\end{enumerate}

\item What is the typical number of transactions made in one of your groups in a given day? Someone paying for everyone’s lunch counts as 1 transaction.
\begin{enumerate}
\item $<=1$
\item $2-3$
\item $4-5$
\item $6-9$
\item $10+$
\end{enumerate}

\item Is there a particular time of day where you use the app most?
\begin{enumerate}
\item Morning
\item Midday
\item Evening
\item Weekend
\item No Particular Time
\end{enumerate}

\item When charging your friends for payments, to what degree do you usually round the cost?
\begin{enumerate}
\item Nearest cent (no rounding)
\item Nearest 5 cents
\item Nearest 10 cents 
\item Nearest 50 cents
\item Nearest \$1
\item Other (please specify)
\end{enumerate}

\item How often do you reject charges on this app?
\begin{enumerate}
\item I never have
\item Very rarely (less than Monthly)
\item Monthly
\item Weekly
\item Daily
\end{enumerate}

\item What are some issues you see as barriers to use for existing payment splitting apps? Please select all that apply.
\begin{enumerate}
\item Peers not using them
\item Slowness or crashing
\item Lack of privacy of personal finance/behavior data from app provider
%the most amazing bug ever
\item Poor usability: difficult to form groups, enter charges, settle balances, etc.
\item None: payment splitting apps are fine as-is
\item Other (please specify) 
\end{enumerate}

\item If your app has a social media component, do you set your transactions to “Public,” “Friends Only,” or “Participants Only”? If it does not, which would you choose if it did?
\begin{enumerate}
\item Public
\item Friends Only
\item Participants Only
\end{enumerate}

\item Suppose your payment splitting app is hacked and all the information in it is stolen. Please rank the following in order of how sensitive the information is to you, with 1 being most sensitive.
\begin{enumerate}
\item Businesses/locations where transactions take place
\item Dollar amounts of transactions
\item Knowing who has rejected a charge and which charges they rejected
\item Knowing who is in the group
\item Knowing who is involved in each transaction and who pays
\item The time when each transaction took place
\end{enumerate}

\end{enumerate}

\section{Full Semihonest Construction}\label{semihonestconstruction}
Here is a formal description of the version of our scheme providing security only against a semihonest server.

Our semihonest secure payment splitting scheme $\mathcal{P}_p$ with security parameter $\lambda$ and group size $N$ uses a PRF $f$ with range $\{0,1\}^\lambda$. After setup, the scheme proceeds in rounds, and messages are sent from each client at a rate of one per round. Let $m$ at any time denote the round number at which the current operation began, and let $r_{m,i,j}$ denote $f(\textsf{sk}, m||i||j)$ except $r_{m,i,j}=0$ for values of $m$ smaller than the round in which user $\mathcal{U}_i$ joined the group or after it left.

\noindent\textbf{Setup}. The server stores a vector \textsf{b} of length $N$ constructed from copies of $0$.  Let $\textbf{a}$ represent a vector of length $N$ constructed from $N$ copies of $1$. The users store a value $b_i=0$ and the key $\textsf{sk}$. 

\noindent\textbf{Request}. User $i$ requests a unit of currency from user $j$ according to the following steps. In the next round $m$, user $\mathcal{U}_i$ creates a vector $\textbf{v}_i$ where $\textbf{v}_{i,j}=1+r_{m,i,j}$ and $\textbf{v}_{i,k}=0+r_{m,i,k}$, $\forall k\neq j$. All other users $\mathcal{U}_k$ create vectors $\textbf{v}_k$ where $\textbf{v}_{k,k}=1+r_{m,k,k}$ and $\textbf{v}_{k,k'}=0+r_{m,k,k'}$, $\forall k'\neq k$. Each user sends its vector $\textbf{v}_i$ or $\textbf{v}_k$ to the server.

Upon receiving the messages from clients for the round, the server takes the sum $\textbf{v}=\Sigma_{i=1}^N\textbf{v}_i$ and also sums the values in \textbf{v} to get $v'$. The server then sets $\textbf{b}=\textbf{b}+\textbf{v}-\textbf{a}$ and computes the value $c=\Sigma_{i=1}^N (\textbf{v}_{i,i}-1)\cdot 2^i$. The server sends the tuple $(v', c, b_l)$ to user $\mathcal{U}_l$, $l\in [N]$. 

Each user receives the values $(v'^*, c^*, b_l^*)$ from the server. Then there are a number of cases:

\begin{enumerate}
\item If $v'^*\neq N+\Sigma_{i=0}^N\Sigma_{j=0}^Nr_{m,i,j}$, the user sends an error message to the server who sets $\textbf{b}=\textbf{b}-\textbf{v}+\textbf{a}$ and outputs $\bot$. 

\item Otherwise, if $c^*=-2^i-\Sigma_{j=1}^N(r_{m,j,j}\cdot2^j)$ for some $i\in[N]$, each user $\mathcal{U}_l$ sets $b_l=b_l^*$. In this case, user $i$ outputs $\bot$ if $v_{i,i}=1$ (if it did not make a request).  

\item If neither of the above cases apply, then more than one request collided in the same round. Let the vector $\textbf{c}$ consist of the values $i_1, ..., i_N$ such that $c^*=\Sigma_{j=1}^N( -2^{i_j}-r_{m,j,j})$. In the next round ($m+1$), each party $\mathcal{U}_i, i\in[N]$ submits a vector $\textbf{v}_i$ such that $\textbf{v}_{i,i}=b_i^*-b_i-r_{m,i,i}+r_{m+1,i,i}$ and $\textbf{v}_{i,k}=0+r_{m+1,i,k} \forall k \neq i$ and the server behaves as above. Then the various requests that collided in this round are repeated, each in a subsequent round, in order from smallest to largest value of $i_j$ as the requester. 
\end{enumerate}

\noindent\textbf{Trace}. User $j$ checks if anyone has charged her in round $m^*$ as follows. First she computes $v^*=b_j^*-b_j-\Sigma_{l=1}^Nr_{m^*,l,j}$ for the values of $b_j$ and $b_j^*$ before/after the round $m^*$ in question to see if she has been charged at all. If her balance has shrunk as a result of the transaction, she looks at the value of $c^*$ in that round and sets $i^*=\log_2(-c^{*} - \Sigma_{l=1}^N(r_{m^*,l,l}\cdot2^l))$. 

\noindent\textbf{Settle}. The server outputs the vector \textbf{b} to the users. The output vector for users is formed by setting $b_i=\textbf{b}_i-\Sigma_{m'=1}^m\Sigma_{j=1}^Nr_{m,i,j}$ for each entry $\textbf{b}_i\in\textbf{b}$.

\section{Deferred Proofs}\label{proofsAppendix}

%We now prove the theorems whose prooofs were deferred in Section~\ref{security}. %We begin by proving security for the semihonest scheme, and then move on to the fully malicious scheme, pointing out which parts of the proof of the semihonest scheme must be changed. 

\begin{theorem}\label{semihonest}
Assuming $f$ is a secure PRF, the semihonest secure payment splitting scheme $\mathcal{P}_p$ has correctness, server privacy, debtor privacy, and user integrity. 
\end{theorem}
\begin{proof}
Correctness follows from the construction, so we do not discuss it further. For server privacy, it's important that the protocol proceeds in fixed rounds and that the server does not know when a collision has happened or is being resolved, so the server just gets a vector of masked values from each user in each round. The proof will rely on the fact that the PRF outputs we use to mask values are indistinguishable from random, meaning that an adversary can never tell what masked values it has received. 

\begin{lemma}
Assuming $f$ is a secure PRF, the semihonest secure payment splitting scheme $\mathcal{P}_p$ has server privacy. 
\end{lemma}

\begin{proof}[Proof (server privacy)]
We proceed by a series of indistinguishable hybrids beginning with the experiment \textsf{PRIV[$\mathcal{A},\lambda,0$]} and ending with \textsf{PRIV[$\mathcal{A},\lambda,1$]}. Let \textsf{PRFADV} be the advantage of adversary $\mathcal{A}$ (playing the role of the server) in distinguishing an output of $f$ from a random string. The list of hybrids is as follows:
\begin{itemize}
\item[--] $\mathsf{H_0}$: The real privacy experiment, \textsf{PRIV[$\mathcal{A},\lambda,0$]}

\item[--] $\mathsf{H_1}$: Same as the previous hybrid, but the outputs of $f(\textsf{sk}, \cdot)$ are replaced by outputs of a random function. This is indistinguishable from the previous hybrid by the PRF security of $f$.

\item[--] $\mathsf{H_2}$: Same as the previous hybrid, but the transaction executed by $\mathcal{C}$ is $t_1$ instead of $t_0$.

\item[--] $\mathsf{H_3}$: Same as the previous hybrid, but the random function outputs used to mask each value sent to the server are replaced by evaluations $f(\textsf{sk}, \cdot)$ as described in the construction of the scheme $\mathcal{P}_p$. This is indistinguishable from the previous hybrid by the PRF security of $f$ and is exactly the security experiment \textsf{PRIV[$\mathcal{A},\lambda,1$]}.
\end{itemize}

We use a standard argument to show that adversary $\mathcal{A}$ distinguishes between experiments $\mathsf{H_0}$ and $\mathsf{H_1}$ with advantage at most \textsf{PRFADV}. The argument for experiments $\mathsf{H_2}$ and $\mathsf{H_3}$ is the same, so we omit it.

We use the adversary $\mathcal{A}$ that distinguishes between the outputs of $\mathsf{H_0}$ and $\mathsf{H_1}$ to construct an adversary $\mathcal{B}$ that wins the PRF security game for $f$ with the same advantage. $\mathcal{B}$ acts as the challenger in the privacy game with $\mathcal{A}$ while simultaneously playing as the adversary in the PRF security game. It reproduces the game for $\mathsf{H_0}$ exactly except that any queries to $f(\textsf{sk}, \cdot)$ are replaced by queries to the PRF security game challenger. Observe that if the PRF challenger is using a PRF on a randomly sampled key, then $\mathsf{B}$ provides a perfect simulation of the game $\mathsf{H_0}$. On the other hand, if the PRF challenger is using a random function, $\mathsf{B}$ provides a perfect simulation of the game $\mathsf{H_1}$. Thus the output of $\mathcal{A}$ wins the PRF security game with the same probability that it distinguishes between $\mathsf{H_0}$ and $\mathsf{H_1}$.

Observe that since values sent to the server in $\mathsf{H_1}$ and $\mathsf{H_2}$ are masked with independently random strings, the distributions of messages sent to the server in these two worlds is identical. As such, the advantage of adversary $\mathcal{A}$ in distinguishing between \textsf{PRIV[$\mathcal{A},\lambda,0$]} and \textsf{PRIV[$\mathcal{A},\lambda,1$]} is at most $2\cdot \textsf{PRFADV}=2\cdot\textsf{negl}(\lambda)=\textsf{negl}(\lambda)$, completing the proof. 
\end{proof}
 
Next, we prove that our scheme satisfies the definition of debtor privacy. Note that because the debtor privacy adversary can corrupt users in the group, it has access to the group key. As such, debtor privacy will not use the security of the PRF $f$, relying instead only on the structure of our protocol.
 
 \begin{lemma}
 The semihonest secure payment splitting scheme $\mathcal{P}_p$ has debtor privacy.
 \end{lemma}
 
\begin{proof}[Proof (debtor privacy)]
We will directly prove debtor by privacy by showing that the view of the adversary in the games \textsf{DEBTPRIV[$\mathcal{A},\lambda,0$]} and \textsf{DEBTPRIV[$\mathcal{A},\lambda,1$]} are distributed identically. The view of each user controlled by the adversary in each round of the debtor privacy game consists of an updated masked balance (which will be different for each user), a check value $v'$ corresponding the sum of all entries in every user's input vector for that round, and the tracing value $c$ used to find out who sent a charge in that round. In a settle operation, each user gets an identical vector $\textbf{b}$. We will show that all of these values are independent of $b$. 

Note that the balances of corrupted users are never affected differently in transactions $t_0$ and $t_1$ (or else $\mathcal{C}$ outputs 0), so the balances of the corrupted users will always be the same at all points regardless of $b$. Thus the balances of corrupted users after a transaction where $t_0$ and $t_1$ differ will be the same as they were before with some new masking values added in. These masking values do not depend on $b$ (they only depend on the round number and the PRF key), so malicious users' balances are distributed identically in \textsf{DEBTPRIV[$\mathcal{A},\lambda,0$]} and \textsf{DEBTPRIV[$\mathcal{A},\lambda,1$]}. 

Next, since honest users always submit vectors whose entries sum to 1 in every round, the contribution of every honest user to $v'$ will be independent of~$b$. 

The tracing value $c$ depends only on the entry each user's vector places in the user's own index. An honest user who is making a charge will put a masked 0 in this position, but all other honest users will put a masked 1 because they will not be making any charges. 
Thus the input of a user being charged and a user not being charged (with the exception of the one making the charge) will always be the same value. 
This means that the tracing value $c$ will be distributed the same regardless of the value of~$b$. 

The security definition requires that all balances be identical across both transcripts when a settlement occurs (or else $\mathcal{C}$ outputs 0), so the vectors \textbf{b} sent by the server during settlement will naturally be the same regardless of~$b$. 

Since we have shown that every element of the adversary's view is distributed identically in \textsf{DEBTPRIV[$\mathcal{A},\lambda,0$]} and \textsf{DEBTPRIV[$\mathcal{A},\lambda,1$]}, we can conclude that no PPT adversary can distinguish between the two experiments with any advantage. 

\end{proof}

%We now prove user integrity, completing the proof of theorem~\ref{semihonest}.

\begin{lemma}
The semihonest secure payment splitting scheme $\mathcal{P}_p$ has user integrity. 
\end{lemma}
\begin{proof}[Proof (user integrity)]
To prove user integrity, we need to prove three separate claims. First, we must show that any input that would cause the sum of user balances (if the group were to settle) to sum to a nonzero value must be detected, even if all but one member of the group is controlled by a malicious adversary. Second, we must show that a user's locally held balance always matches the balance reported when settling. Third, we must show that an honest user who did not initiate a transaction in a given round can always \emph{detect} if he has been framed as having done so. 

The first component of user integrity is satisfied by the server polling all users to ensure that the sum of their inputs equals the (masked) group size $N$. Since all users are polled, as long as one honest user remains to catch an error, the server will learn that a user has violated integrity in that round. This suffices to prove the desired property because all balances are set to zero at the time of a group's creation, so the only opportunity for the balances to sum to a nonzero value when a group settles is if there is a round where the sum of all the users' inputs is not equal to the group size $N$ (the sum is $N$ instead of 0 because the server subtracts off 1 from each index before adding the sum of user inputs into the balances). The sum of all users' inputs will equal $N$ in an honest round by the correctness of the protocol. 

The fact that a user's balance always matches the output of settling follows directly from the construction because the server sends each user its current entry in the balance vector in each round. 

Proving that an honest user can detect whether he has been framed follows from the correctness of the construction. By this we mean that the security property does not rely on any particular property of the tracing mechanism except that users will agree on who \emph{appears} to be making a charge according to the rules of the protocol. Our protocol, when executed honestly, always results in the tracing value $c$ containing a (masked) power of two $-2^i$ corresponding to a charge coming from user $i$. As such, whenever $c$ contains $-2^i$ for a user $i$ who did not make a charge in the corresponding round, user $i$ can tell that he is being framed. 
\end{proof}
\end{proof}

\begin{theorem}\label{malicious}
Assuming $f$ is a secure PRF, the fully malicious secure payment splitting scheme $\mathcal{P}_m$ has correctness, server privacy, debtor privacy, user integrity, and server integrity (against a malicious server). 
\end{theorem}

Proofs of correctness, server privacy, debtor privacy, and user integrity for the malicious secure scheme $\mathcal{P}_m$ closely resemble those of the semihonest scheme $\mathcal{P}_p$. The server privacy proof is identical because replacing the various $r_{m,i,j}$ values in the construction with random values suffices to render the games \textsf{PRIV[$\mathcal{A},\lambda,0$]} and \textsf{PRIV[$\mathcal{A},\lambda,1$]} indistinguishable. Proofs of debtor privacy and user integrity only replace the integrity check and balance values seen by the users with the more complex values involving $s$ that appear in the malicious secure scheme. 

Server integrity holds because it is hard for the server to generate a message to a client that will be accepted as legitimate. This is the case because there are only a limited number of values a client will accept in a given round, and the server can do no better than guessing at correct ones. For example, a client knows that its balance will either be incremented or decremented, and that the charger will be one of $N$ possible parties. Since the pseudorandom value $s$ and the per-message randomness mask messages from the server, the server has a negligible probability of finding a message that can be correctly unmasked to an acceptable multiple of $s$.

\begin{lemma}\label{malicious}
Assuming $f$ is a secure PRF, the fully malicious secure payment splitting scheme $\mathcal{P}_m$ has server integrity (against a malicious server).
\end{lemma}

\begin{proof}[Proof (server integrity)]
Let \textsf{PRFADV} be the advantage of adversary $\mathcal{A}$ in distinguishing an output of $f$ from a random string. First, using a hybrid analogous to the step between hybrids $\mathsf{H}_0$ and $\mathsf{H}_1$ in the server privacy proof above, we can replace all evaluations of $f$ with evaluations of a random function.

In each round, the server sends each client an updated balance, an integrity check value, and a value $c$ which identifies anyone making a charge in that round. Whatever value the server produces must be of the form $v^*=s\cdot x + y$ where $s$ and $y$ are known to the users but not the server, with $y$ changing for every message and $s$ fixed. In the case of the integrity value, $x$ must be $N$, in the case of the balance, $x$ must be the previous balance $+/-1$ (unless the requester value indicates a collision, where it doesn't matter what it is), and for the requester value, $x$ must be a power of 2 or a sum of at most $N$ powers of 2. In the case a second round is used for users to send their masked inputs without multiplying by $s$, each value $x$ sent by the server must correspond to the same value sent in the previous round. This means that the malicious server has at most $z=O(N)$ acceptable values that it can set $x$ to be for any message.

Since each value sent from the server to the users in each round is a distinct combination of a subset of users' inputs, $s$ and $y$ are independently random. 
As a result, the probability that $v^*-y = s \cdot x$ for a value of $x$ expected/acceptable to a client is $\frac{1}{2^\lambda}$ for each acceptable value. Since there are $O(N)$ acceptable values, there will be an $O(\frac{N}{2^\lambda})$ probability that a server successfully forges a given message. Taking a union bound over all server messages in the course of the protocol ($m$ rounds of $O(N)$ messages each), we get that $\mathcal{A}$ can break server integrity with probability at most $O(\frac{mN^2}{2^\lambda})$. 

We have now covered server integrity for requests and tracing (both are included in the round by round structure of our protocol), but we must cover settling separately. Fortunately, settling operates according to more or less the same principles as the round protocol. The server sends each user a value corresponding to every other user's balance $p_i$ masked with an independently random value $y_i$ followed by a second value $v^*$ which must correspond to $s\cdot p_i +y_i'$ for another independently random $y_i'$. As above, the probability of successfully fooling a user into accepting an incorrect value $p_i^*$ is $\frac{1}{2^\lambda}$, and the union bound over all messages sent to all users for each of $n$ settlements is less than $\frac{nN^2}{2^\lambda}$.

Finally, the probability that a message from a malicious server $\mathcal{S}^*$ deviates from what an honest server $\mathcal{S}$ would send without being caught is at most the sum of the distinguishing advantage between $\mathsf{H_0}$ and $\mathsf{H_1}$, and the probabilities of failure for the rounds and for settlement. That is $\textsf{PRFADV}+O(\frac{mN^2}{2^\lambda})+\frac{nN^2}{2^\lambda}<\textsf{negl}(\lambda)$. 
\end{proof}

\end{document}